\documentclass[11pt,a4paper,notitlepage]{article}

\usepackage{amsmath}        
\usepackage{amsfonts}       
\usepackage{amsthm}         
\usepackage{bbding}         
\usepackage{bm}             
\usepackage{graphicx}       
\usepackage[square,sort,comma,numbers]{natbib}         
\usepackage[nottoc]{tocbibind} 
\usepackage{paralist}       
\usepackage[usenames]{xcolor}  

\usepackage{xfrac}
\usepackage{tikz}
\usepackage{tikz-cd}

\usetikzlibrary{arrows,shapes,positioning}
\usetikzlibrary{decorations.markings}

\usetikzlibrary{arrows, matrix}

\usepackage{wrapfig}
\usepackage{caption}
\usepackage{titlesec}
\usepackage{filecontents}
\usepackage{mathtools}

\usepackage{mdwlist}
\usepackage{mathrsfs}
\usepackage[shortlabels]{enumitem}
\usepackage{amssymb}
\usepackage{textcomp}
\usepackage[toc,page]{appendix}

\usepackage{pgfplots}
\pgfplotsset{compat=1.7}

\usepackage[utf8]{inputenc} 
\usepackage{authblk} 

\usepackage[cal=boondox,calscaled=1.0]{mathalfa}

\usepackage{physics}

\numberwithin{equation}{section}
\numberwithin{figure}{section}
\theoremstyle{definition}
\newtheorem{definition}{Definition}[section]
\newtheorem{observation}{Observation}[section]

\newtheorem{theorem}{Theorem}[section]

\newtheorem{example}{Example}[section]

\usepackage{natbib}
\usepackage[nottoc]{tocbibind}

\title{Discrete Linear Canonical Evolution}
\author[1]{Jakub Káninský}

\affil[1]{Charles University, Faculty of Mathematics and Physics, Institute of Theoretical Physics. E-mail address: jakubkaninsky@seznam.cz}

\date{\today}
\titleformat{\section}{\normalfont\scshape\large}{\thesection}{2em}{}
\titleformat{\subsection}{\normalfont\scshape\normalsize}{\thesubsection}{1em}{}
\titleformat{\subsubsection}{\itshape\normalsize}{\thesubsubsection}{1em}{}

\renewenvironment{abstract}
 {\small
  \begin{center}
  \end{center}
  \list{}{
    \setlength{\leftmargin}{15mm}%
    \setlength{\rightmargin}{\leftmargin}%
  }
  \item\relax}
 {\endlist}

\usepackage{tocloft}
\begin{document}

\maketitle

\begin{abstract}
This work builds on an existing model of discrete canonical evolution and applies it to the general case of a linear dynamical system, i.e., a finite-dimensional system with configuration space isomorphic to $ \mathbb{R}^{q} $ and linear equations of motion. The system is assumed to evolve in discrete time steps. The most distinctive feature of the model is that the equations of motion can be irregular. After an analysis of the arising constraints and the symplectic form, we introduce adjusted coordinates on the phase space which uncover its internal structure and result in a trivial form of the Hamiltonian evolution map. For illustration, the formalism is applied to the example of massless scalar field on a two-dimensional spacetime lattice.
\end{abstract}

\vspace{3\baselineskip}

\tableofcontents

\newpage

\section{Introduction}

Discrete dynamical systems have wide use in many areas of science including physics, engineering, biology, demography, finance and economics \cite{Zhang2020, Zhang2006, Allen2004}. Sometimes this is because the problem at hand is naturally described in the discrete setting, sometimes rather because such formulation is far easier and better suited for an automated computation. Over the past decades, the latter reason have given rise to a vast body of mathematics and computer science concerned with discretization \cite{Zienkiewicz2013}. In the most general regard, we can define a discrete dynamical system to be any system whose evolution happens in a series of discrete time steps. This is often realized by iterations of a single fixed function on the phase space; where the function may be linear or non-linear in the phase-space coordinates and the phase space may have one or more dimensions \cite{Galor2007}. However, one may consider a slightly richer setting, in which one allows the evolution to be governed by a set of parameters which change in time.\\

In this work we study discrete dynamical systems of arbitrary finite dimension with linear evolution mappings and time-varying parameters. In doing so, we adopt an approach which is profoundly physical. The evolution of our system will be defined by an action functional, which gives rise to equations of motion in the canonical, Hamiltonian language, and a phase space endowed with canonical coordinates as well as the symplectic form. The \textit{discrete linear canonical evolution} advertised by the title is meant in exactly this context: as a classical discrete Hamiltonian evolution of a linear dynamical system well known from physics. Our interest in this particular setting is motivated by discrete models of spacetime which represent an important complement to the standard formulation of general relativity based on continuous differential manifolds. These can be used to study the behavior of classical and quantum fields on simplicial manifolds \cite{Hamber2009, McDonald2010, Sorkin1975, Foster2004, Brower2016, Brower2018}. Discrete spacetime models have also proven powerful for studying certain aspects of gravitation and are central to a few well-established approaches to quantum gravity \cite{Oriti2009, Hamber2009, Jha2018, Mikovic2018, Loll2019, Dittrich2011}.\\

Our starting point is the model of \textit{discrete canonical evolution} introduced by B.~Dittrich and P.~A.~H\"{o}hn in a series of articles \cite{Dittrich2011, Dittrich2012, Dittrich2013} motivated by an application to simplicial gravity. Among other things, the authors provide two versions of the discrete canonical evolution, a global and a local one, and a key theorem about the conservation of the symplectic form. Notably, the formalism allows for irregular evolution which need not always provide one-to-one correspondence between initial and final states. Thanks to that, it extends to systems with degenerate action or time-varying number of degrees of freedom. The irregularity induces constraints as well as non-uniqueness of the classical evolution: a feature which is otherwise generally not very common in physics. Consequently, conservation of the symplectic form is only limited. The authors give an analysis of the arising constraints \cite{Dittrich2013} which follows the path of the original Dirac's classification \cite{Dirac2001}.\\

The problematics of discrete linear evolution was subsequently addressed in the article \cite{Hoehn2014} where the author performed a detailed classification of constraints for this special case. Therein, both constraints and degrees of freedom were classified into eight types according to their dynamical behavior, and a new basis on the phase space was defined which separates first class and second class constraints. A notion of \textit{reduced phase space} was introduced, followed by an analysis and classification of so-called \textit{effective actions}. All these concepts are very useful in helping us understand the Hamiltonian evolution of discrete linear systems. The present work has the same aim, it only takes a different path in pursuing it. Instead of constraints, we put the symplectic structure into the foreground. We also limit our analysis to a single time-step, which allows for a much simpler viewpoint. There are number of places where the parallel between our work and \cite{Hoehn2014} becomes significant, these will be pointed out in the text. Eventually, let us mention that the article \cite{Hoehn2014} treats in its last section also the quantum case which is beyond the scope of the present work.\\

Within the paper, we only consider the global version of the discrete canonical evolution. Its brief review can be found in Sec. 2. It will be applied strictly to the special case of \textit{linear dynamical system}: a system of finite dimension, vector configuration space and linear equations of motion. On one hand, the assumption of linearity is very restrictive; on the other hand, it allows us to analyze the evolution efficiently by means of standard linear-algebraic tools. We recall the definition of linear dynamical system along with a couple of important details in Sec. 3. The only thing we alter about this classical definition is that we exchange the (implicit assumption of) continuous time for discrete time steps. The next Sec. 4 summarizes some essentials from linear algebra. We will take advantage of these throughout the work.\\

Sec. 5 contains the main body of the work. As in \cite{Hoehn2014}, we use the assumption of linearity to further develop the formalism given in the original papers \cite{Dittrich2012, Dittrich2013} and explore the impact of irregularity on the dynamics of the system. However, unlike in \cite{Hoehn2014}, we employ singular value decomposition to describe the Hamiltonian evolution map and all the constraints explicitly. We provide an elementary rewriting of the linear canonical evolution in terms of matrices, followed by an analysis of constraints with respect to the symplectic structure. We further take the opportunity to build two special coordinate frames on the phase space which are fundamentally different from the frame described in \cite{Hoehn2014}. These are adapted to the constraint surfaces of the time-step in question and result in a trivial evolution prescription. Both the previously existing and newly introduced notions are given an explicit matrix form, and together constitute an effective framework well suited for an immediate implementation. In the final part of the section, we discuss global solutions.\\

Eventually, in Sec. 6 we provide a simple yet physically sound illustration of the theoretical concepts introduced before. In order to do that, we consider a particular instance of a discrete linear system: massless scalar field on a fixed two-dimensional spacetime lattice. Similar models \cite{McDonald2010, Hamber2009, Brower2018} have been studied since the pioneering work by Regge \cite{Regge1961}. In our classical case, the formalism of discrete linear evolution applies straightforwardly. After a brief general description of the system, we offer a few mini-examples of time-slice lattices and work out the corresponding evolution moves in detail. It is shown how the dynamics of the field is shaped by the geometry and causal structure of the underlying spacetime, here described solemnly by its lattice representation.\\

\vspace{\baselineskip}

\section{Discrete Canonical Evolution}

For a treatment of the canonical evolution of a discrete system we refer to the formalism originally built to describe the evolution of simplicial gravity \cite{Dittrich2011, Dittrich2012, Dittrich2013}. Its functioning is briefly reviewed in this section. For more details, we suggest the reader consult the original articles.\\

Let $ \mathcal{Q}_n $ be the configuration space of a discrete system at time-slice $ n $ with coordinates $ x_{n A} $ (we will occasionally omit the index $ A $). Note that the configuration spaces at different time-slices need not be of the same dimension $ \dim(\mathcal{Q}_n) \equiv q_{n} $. The dynamics of the system shall be described by the action
\begin{equation}\label{action}
S = \sum_{n=0}^{t-1} S_{n+1}(x_{n},x_{n+1})
\end{equation}
where the sum ranges over the individual time-slices $ n $. The action contribution $ S_{n+1} $ governs the discrete time evolution move during the time-step between $ n $ and $ n+1 $. We assume that the action is additive, so that the sum in \eqref{action} makes sense. (This is a restrictive condition since one could in principle consider non-additive actions; for instance those that include interactions between more than two consecutive time-slices. These are ruled out by our additivity requirement.)\\

Let us treat \emph{global} time evolution moves, i.e., such moves that each of the variables at a given time-slice is involved in the move and only occurs at this one time-slice, so that neighboring time-slices do not overlap. For example, evolution moves in simplicial gravity which evolve between disjoint spacial hypersurfaces are global.\\

Consider three consecutive time-slices $ n-1, n, n+1 $ and the boundary value problem defined by the data at times $ n-1 $ and $ n+1 $. That is, we are given the boundary data $ x_{n-1} $ and $ x_{n+1} $ and ought to extremize $ S = S_n(x_{n-1},x_{n}) + S_{n+1}(x_{n},x_{n+1}) $ with respect to $ x_n $. This yields the equations of motion
\begin{equation}
0 = \frac{\partial S_n}{\partial x_n} + \frac{\partial S_{n+1}}{\partial x_n}
\end{equation}
which may or may not be uniquely solvable for $ x_n $ as a function of $ x_{n-1}, x_{n+1} $, depending on whether the system under consideration is regular or irregular. An initial value problem can be treated in analogy by computing $ x_{n+1} $ from $ x_{n-1},x_{n} $; then the equations of motion provide the Lagrangian time evolution $ \mathbb{L}_n : \mathcal{Q}_{n-1} \times \mathcal{Q}_{n} \rightarrow \mathcal{Q}_{n} \times \mathcal{Q}_{n+1} $. It may not, however, be defined on all of $ \mathcal{Q}_{n-1} \times \mathcal{Q}_{n} $, nor map to all of $ \mathcal{Q}_{n} \times \mathcal{Q}_{n+1} $, nor be unique in the presence of constraints.\\

\vspace{\baselineskip}

In order to describe the dynamics in canonical language, one may introduce \emph{discrete Legendre transformations}. For an arbitrary time-slice $ n $, we have $ S_n : \mathcal{Q}_{n-1} \times \mathcal{Q}_{n} \rightarrow \mathbb{R} $ where $ \mathcal{Q}_{n-1} \times \mathcal{Q}_{n} $ is a fibre bundle. Pick a point $ q_{n-1} $ in $ \mathcal{Q}_{n-1} $. We denote the fibre over $ q_{n-1} $ by $ \mathcal{F}_{n} (q_{n-1}) $. Notice that $  \mathcal{F}_{n} \cong \mathcal{Q}_{n} $. Choose a point $ f_n $ in $ \mathcal{F}_{n} $ and a curve $ \gamma(\varepsilon) $ in $ \mathcal{F}_{n} $ such that $ \gamma(0) = f_n $. This allows us to define the \emph{post-Legendre transform} $ \mathbb{F}^{+} S_n : \mathcal{Q}_{n-1} \times \mathcal{Q}_{n} \rightarrow T^{\ast} \mathcal{Q}_n $ by
\begin{equation}
\mathbb{F}^{+} S_n (f_n) \cdot \left. \frac{d \gamma(\varepsilon)}{d \varepsilon} \right\rvert_{\varepsilon = 0} = \left. \frac{d S_n(\gamma(\varepsilon)) }{d \varepsilon} \right\rvert_{\varepsilon = 0}
\end{equation}
Here it should be understood that $ \left. \frac{d \gamma(\varepsilon)}{d \varepsilon} \right\rvert_{\varepsilon = 0} $ is a vector in the tangent space $ T_{f_{n}} \mathcal{Q}_n $, and is being contracted with the covector $ \mathbb{F}^{+} S_n (f_n) $ belonging to the cotangent space $ T^{\ast}_{f_{n}} \mathcal{Q}_n $. The point $ q_{n-1} $ is implicit in the equation, entering both the $ f_{n} $ and $ \gamma(\varepsilon) $. Now exchange the roles of $ \mathcal{Q}_{n-1} $ and $ \mathcal{Q}_{n} $ and choose $ f_{n-1} $ in $ \mathcal{F}_{n-1}(q_{n}) $. Let $ \eta(\varepsilon) $ be a curve in $ \mathcal{F}_{n-1} $ such that $ \eta(0) = f_{n-1} $. In analogy, we define the \emph{pre-Legendre transform} $ \mathbb{F}^{-} S_n : \mathcal{Q}_{n-1} \times \mathcal{Q}_{n} \rightarrow T^{\ast} \mathcal{Q}_{n-1} $ by
\begin{equation}
\mathbb{F}^{-} S_n (f_n) \cdot \left. \frac{d \eta(\varepsilon)}{d \varepsilon} \right\rvert_{\varepsilon = 0} = - \left. \frac{d S_n(\eta(\varepsilon)) }{d \varepsilon} \right\rvert_{\varepsilon = 0}
\end{equation}

The cotangent bundles $ \mathcal{P}_{n-1} = T^{\ast} \mathcal{Q}_{n-1} $ and $ \mathcal{P}_{n} = T^{\ast} \mathcal{Q}_{n} $ are the phase spaces upon which we base the canonical formalism. The transforms are depicted in the following diagram.
\begin{equation*}
\begin{tikzcd}[column sep=0.2pc,row sep=0.5pc]
\mathcal{Q}_{n} \arrow[dashrightarrow]{rd} & & \mathcal{P}_{n} = T^{\ast} \mathcal{Q}_{n} \\
& \mathcal{Q}_{n-1} \times \mathcal{Q}_{n}  \arrow{ru}{\mathbb{F}^{+} S_n} \arrow{rd}{\mathbb{F}^{-} S_n} & \\
\mathcal{Q}_{n-1} \arrow[dashrightarrow]{ru} & & \mathcal{P}_{n-1} = T^{\ast} \mathcal{Q}_{n-1} \\
\end{tikzcd}
\end{equation*}
The coordinate versions of the post- and pre-Legendre transformations are
\begin{equation}\label{Legtracoord}
\begin{aligned}
\mathbb{F}^{+} S_{n} (x_{n-1}, x_{n}) & = ( x_{n}, ^{+}p_{n} ) & \qq{ where } & ^{+}p_{n} = \frac{\partial S_{n}}{\partial x_{n}} \\
\mathbb{F}^{-} S_{n} (x_{n-1}, x_{n}) & = ( x_{n-1}, ^{-}p_{n-1} ) & \qq{ where } & ^{-}p_{n-1} = - \frac{\partial S_{n}}{\partial x_{n-1}}
\end{aligned}
\end{equation}
Note that in context of discrete evolution, the time-step action contribution $ S_{n+1} $ has the role of Lagrangian. Then, its variation enables one to define so-called \textit{Lagrange two-forms}. They are given as
\begin{equation}
\Omega_{n+1}(x_{n},x_{n+1}) = - \frac{\partial^{2} S_{n+1}}{\partial x_{n A} \partial x_{n+1 B}} d x_{n A} \wedge d x_{n+1 B}
\end{equation}
We remark that this form is not fully analogical to the classical Lagrange two-form, see e.g. Chap. 7 of \cite{Marsden2010}, since it is built purely out of coordinates and does not contain velocities (at least not explicitly). This is convenient, since there is no canonical notion of velocity to start with. One can check that the Lagrange two-forms arise from pulling back the canonical two-forms
\begin{equation}\label{cantwoform}
\omega_{n} = \sum_{A = 1}^{q_{n}} d x_{n A} \wedge d p_{n A}
\end{equation}
with the Legendre transformation,
\begin{equation}
\Omega_{n} = \mathbb{F}^{+} S_{n}^{\ast} \omega_{n}, \qquad \Omega_{n+1} = \mathbb{F}^{-} S_{n+1}^{\ast} \omega_{n}
\end{equation}
In \eqref{cantwoform}, $ p_{n} $ stands either for $ p_{n}^{+} $ (in case we use $ \mathbb{F}^{+} S_{n}^{\ast} $) or $ p_{n}^{-} $ (in case we use $ \mathbb{F}^{-} S_{n+1}^{\ast} $).\\

Now that we have passed to the canonical language, we may go on and use it to give a formulation of the equations of motion. To do that, we only need the above given equations
\begin{equation}\label{prepostmoneta}
^{-}p_{n-1} = - \frac{\partial S_{n}}{\partial x_{n-1}}, \qquad ^{+}p_{n} = \frac{\partial S_{n}}{\partial x_{n}} 
\end{equation}
We shall refer to $ ^{-}p $ as \emph{pre-momenta} and to $ ^{+}p $ as \emph{post-momenta}. By virtue of the equations \eqref{prepostmoneta}, we implicitly define the global one-step Hamiltonian evolution map $ \mathbb{H}_{n}: \mathcal{P}_{n-1} \rightarrow \mathcal{P}_{n} $, acting as
\begin{equation}\label{H}
\mathbb{H}_{n} ( x_{n-1}, ^{-}p_{n-1} ) = ( x_{n}, ^{+}p_{n} )
\end{equation}
In practice, one can use the equation for pre-momenta in \eqref{prepostmoneta} to determine $ x_{n} $ (if possible) and use this result in the post-momenta equation in \eqref{prepostmoneta} to determine $ ^{+}p_{n} $. The definition of $ \mathbb{H}_{n} $ does not require equations of motion. However, the equations of motion are present in the canonical picture as the \emph{momentum-matching} of pre- and post-momenta,
\begin{equation}\label{mommatch}
^{+}p_{n} = ^{-}p_{n}
\end{equation}
We can omit the superindices $ + $ and $ - $, assuming that momentum-matching holds.\\

\vspace{\baselineskip}

In the latter, we will discuss the special case in which the dimensions of the configuration spaces at consecutive steps are equal, $ \dim(\mathcal{Q}_{n})  = \dim(\mathcal{Q}_{n}) = q $, but the system is irregular, meaning that
\begin{equation}
\mathrm{det} \frac{\partial^{2} S_{n+1}}{\partial x_{n} \partial x_{n+1}} = 0
\end{equation}
Then, an equal number $ s_{n} $ of left and right null vectors $ L_{n I} $, $ R_{n+1 J} $, respectively, of the Lagrange two-form arises in an open neighborhood in $ \mathcal{Q}_{n} \times \mathcal{Q}_{n+1} $,
\begin{equation}\label{LR}
L_{nI} \frac{\partial^{2} S_{n+1}}{\partial x_{nI} \partial x_{n+1 J}} = 0, \qquad \frac{\partial^{2} S_{n+1}}{\partial x_{nI} \partial x_{n+1 J}} R_{n+1 J} = 0
\end{equation}
Above, $ I, J = 1, ..., s_{n} $. This case is of interest because it is possible to treat systems with temporally varying numbers of degrees of freedom in this irregular fashion.\\

The definitions of pre- and post-Legendre transforms are directly applicable to singular systems. The rank of both is $ 2q - s_{n} $. The Legendre transformations thus fail to be onto, and their images form submanifolds of dimension $ 2q - s_{n} $ in the phase spaces $ \mathcal{P}_{n} $, $ \mathcal{P}_{n+1} $. We call $ \mathcal{C}^{-}_{n} = \Im( \mathbb{F}^{-} S_{n+1} ) \subset \mathcal{P}_{n} $ the \emph{pre-constraint surface} and $ \mathcal{C}^{+}_{n+1} = \Im( \mathbb{F}^{+} S_{n+1} ) \subset \mathcal{P}_{n+1} $ the \emph{post-constraint surface}.\\

We choose $ s_{n} $ irreducible \emph{pre-constraints} $ ^{-}C_{nL} $, $ L = 1, ..., s_{n} $ to describe $ \mathcal{C}^{-}_{n} $. They are automatically satisfied by the pre-momenta $ ^{-}p_{n} $ (because these arise from the corresponding Legendre transform), but impose (through momentum-matching) non-trivial conditions on post-momenta $ ^{+}p_{n} $ resulting from the previous evolution move between $ n-1 $ and $ n $. Analogically, the $ s_{n} $ irreducible \emph{post-constraints} $ ^{+}C^{R}_{n+1} $, $ R = 1, ..., s_{n} $ are automatically satisfied by the post-momenta $ ^{+}p_{n+1} $ but will provide non-trivial conditions for the pre-momenta $ ^{-}p_{n+1} $ by momentum-matching. The pre- and post-constraint surfaces at a time-step $ n $ generally do not coincide, we thus have to restrict the evolution to $ \mathcal{C}_{n} = \mathcal{C}^{-}_{n} \cap \mathcal{C}^{+}_{n} $ at each time-slice. We remark that the constraints can be subjected to analysis analogical to the standard Dirac's classification \cite{Dirac2001}; the authors provide such analysis in \cite{Dittrich2013}.\\

For singular systems, the Hamiltonian evolution map can be only defined between the corresponding constraint surfaces, that is, $ \mathbb{H}_{n+1}: \mathcal{C}^{-}_{n} \rightarrow \mathcal{C}^{+}_{n+1} $. On the other hand, as a consequence of the right null vectors in \eqref{LR}, we are also no longer able to solve the pre-momentum equation in \eqref{prepostmoneta} uniquely for $ x_{n+1} $. Therefore, we must specify $ s_{n} $ additional \emph{a priori} free parameters $ \lambda_{n+1 R} $ to uniquely determine $ x_{n+1}(x_{n}, ^{-}p_{n}, \lambda_{n+1}) $. The post-momenta $ ^{+}p_{n+1} $ will in general depend on the parameters $ \lambda_{n+1 R} $. Likewise, on account of the left null vectors in \eqref{LR}, we can no longer uniquely express $ x_{n} $ as a function of $ x_{n+1} $ and $ ^{+}p_{n+1} $. Then, if we want to evolve the system backwards, we must specify $ s_{n} $ a posteriori free parameters $ \mu_{n L} $ in order to determine $ x_{n}(x_{n+1}, ^{+}p_{n+1}, \mu_{n}) $. We have one a priori free $ \lambda_{n R} $ per post-constraint and one a posteriori free $ \mu_{n L} $ per pre-constraint. The arbitrariness arising in the form of these free parameters may be, however, later reduced by subsequent constraints at later time-slices.\\

In consequence of restricting the Hamiltonian evolution to the constraint hypersurfaces, it cannot be a symplectic map. Nonetheless, a slightly weaker assertion holds:

\begin{theorem}\label{comegac}
	Let $ \iota_n^{-} : \mathcal{C}_n^{-} \rightarrow \mathcal{P}_n $ and $ \iota_{n+1}^{+} : \mathcal{C}_{n+1}^{+} \rightarrow \mathcal{P}_{n+1} $ be embeddings of the constraint surfaces into the corresponding phase spaces, and $ (\iota_n^{-})^{\ast} $, $ (\iota_{n+1}^{+})^{\ast} $ the associated pullbacks. Then
	\begin{equation}
	(\iota_n^{-})^{\ast} \omega_n = \mathbb{H}_{n+1}^{\ast} (\iota_{n+1}^{+})^{\ast} \omega_{n+1} \label{Hcons}
	\end{equation}
\end{theorem}
\begin{proof}
	See the proof of Theorem 6.1 in \cite{Dittrich2012}.
\end{proof}

\vspace{\baselineskip}

\section{The Notion of Linear Dynamical System}

Before giving a definition of the linear dynamical system, let us remind that for a general classical system of finite dimension, the configuration space $ \mathcal{Q} $ is a $ q $-dimensional manifold, and one may define the phase space $ \mathcal{P} = \mathbf{T}^{\ast}(\mathcal{Q}) $ to be the cotangent bundle over $ \mathcal{Q} $ which is a symplectic manifold of dimension $ 2q $. It is equipped with the symplectic two-form
\begin{equation}\label{key}
\omega = \sum_{A = 1}^{q} d x_{A} \wedge d p_{A}
\end{equation}
where $ (x_{1}, ... , x_{q}) $ are some local coordinates on $ \mathcal{Q} $ and $ (p_{1}, ..., p_{q}) $ are components of the cotangent vectors in the coordinate basis associated with $ (x_{1}, ... , x_{q}) $. So far the general case. Now, borrowing from Chap. 2 of \cite{Wald1994}, we give the anticipated definition.

\begin{definition}
	A \textit{linear dynamical system} satisfies the following two conditions:
	
	\begin{enumerate}[{(1)}]
		\item Its configuration space $ \mathcal{Q} $ has a natural vector space structure. Then we can choose a basis in $ \mathcal{Q} $ and use the basis components of vectors in $ \mathcal{Q} $ to define \textit{linear coordinates} $ (x_{1}, ... , x_{q}) $ globally on $ \mathcal{Q} $. These coordinates then give rise to globally well defined \textit{linear canonical coordinates} $ (x_{1}, ... , x_{q}, p_{1}, ..., p_{q}) $ on $ \mathcal{P} $.
		\item The Hamiltonian $ H $ is a quadratic function on $ \mathcal{P} $ so that the equations of motion $ \dot{x}_{A} = \frac{\partial H}{\partial p_{A}}, ~ \dot{p}_{A} = - \frac{\partial H}{\partial x_{A}} $ are linear in the linear canonical coordinates.
	\end{enumerate}
	
\end{definition}

For simplicity we shall denote the canonical (symplectic) basis of $ \mathcal{P} $ as $ \{ e_{I} \}_{I = 1}^{2q} $ and the linear canonical coordinates of a point $ y \in \mathcal{P} $ by $ (y_{1}, ..., y_{q}, y_{q+1}, ..., y_{2q}) $. Then it is understood that $ y = y_{I} e_{I} $ with an implicit summation over $ I = 1, ..., 2q $. \\

A key consequence of the vector space structure of $ \mathcal{P} $ given by the property (1) is that one may identify the tangent space $ T_{x}\mathcal{P} $ at any point $ x \in \mathcal{P}  $ with $ \mathcal{P} $ itself. Under this identification, the symplectic form $ \omega $ becomes a bilinear function $ \omega: \mathcal{P} \times \mathcal{P} \rightarrow \mathbb{R} $. Furthermore, since the components of the symplectic form are constant in the canonical coordinate basis, the corresponding bilinear map is independent of the choice of point $ x $ used to make the identification. We shall refer to $ \omega: \mathcal{P} \times \mathcal{P} \rightarrow \mathbb{R} $ as the \textit{symplectic structure} on $ \mathcal{P} $. For $ y, u \in \mathcal{P} $, it is given by
\begin{equation}\label{symplstruct}
\omega(y,u) = y_{A} u_{A+q} - y_{A+q} u_{A}
\end{equation}
with an implicit summation over $ A = 1, ..., q $; and we will also take advantage of an alternative form $ \omega(y,u) = y^{T} \sigma u = y_{I} ~ \sigma_{IJ} ~ u_{J} $ with summation over $ I, J = 1, ..., 2q $ and the matrix
\begin{equation}\label{key}
\sigma = \begin{pmatrix}
\phantom{-} 0 & \mathbf{1} \\
- \mathbf{1} & 0
\end{pmatrix}
\end{equation}
Although the phase space $ \mathcal{P} $ is generally a symplectic manifold, we see that in case of a linear dynamical system, it may be viewed as a symplectic vector space $ (\mathcal{P} , \omega) $, i.e., a vector space on which is defined a non-degenerate, antisymmetric, bilinear map $ \omega $.\\

Thanks to the property (2) of the linear dynamical system, the Hamiltonian function $ H $ on $ \mathcal{P} $ takes the form
\begin{equation}
H(t, y) = \textstyle \frac{1}{2} ~ y_{I} ~ K(t)_{IJ} ~ y_{J}
\end{equation}
(again, summing over $ I,J = 1, ..., 2q $) where $ K(t) $ is a symmetric $ 2q \times 2q $ matrix. The Hamilton's equations of motion then turn out as
\begin{equation}\label{eqm}
\dot{y}_{I} = \sigma_{IJ} \frac{\partial H}{\partial y_{J}} = \sigma_{IJ} ~ K_{JK} ~ y_{K}
\end{equation}
now, let $ y(t) $, $ u(t) $ be two solutions of the equations of motion \eqref{eqm} and let
\begin{equation}
s(t) = \omega(y(t),u(t)) = \sigma_{IJ} ~ y(t)_{I} ~ u(t)_{J}
\end{equation}
Then we have
\begin{equation}\label{dots0}
\dot{s} = 0
\end{equation}
Thus, for a linear dynamical system, the \textit{symplectic product} $ s $ of two solutions is conserved. (This result is really a consequence of a much more general fact for nonlinear systems that dynamical evolution defines a canonical transformation on phase space.) Thus, the symplectic structure $ \omega: \mathcal{P} \times \mathcal{P} \rightarrow \mathbb{R} $ gives rise to a natural symplectic structure $ \omega $ on the vector space of solutions $ \mathcal{S} $ to the equations of motion, since the symplectic structure on $ \mathcal{S} $ obtained by identifying $ \mathcal{P} $ and $ \mathcal{S} $ does not depend upon the choice of initial time $ t $ under which the identification is made.\\

\vspace{\baselineskip}

\section{Algebraic Preliminaries}
Let us take a moment to remind a few basic tools from linear algebra which will prove very useful for our work. The following is an adaptation of Section 3 of \cite{Kaninsky2020}. We first look at the singular value decomposition and Moore-Penrose pseudoinverse, as defined in \cite{Golub2012}. Then we shortly mention symplectic matrices.\\

\begin{theorem}
	Let $ A \in \mathbb{R}^{m \times n}  $ be an $ m \times n $ matrix with $ m \geq n $. Then there exist orthogonal matrices $ U \in \mathbb{R}^{m \times m} $ and $ V \in \mathbb{R}^{n \times n} $ and a matrix $ \Sigma = \begin{pmatrix}
	\text{diag}(\sigma_{1}, ..., \sigma_{n}) \\
	0
	\end{pmatrix} \in \mathbb{R}^{m \times n} $ with $ \sigma_{1} \geq \sigma_{2} \geq ... \geq \sigma_{n} \geq 0 $, such that
	\begin{equation}\label{key}
	A = U \Sigma V^{T}
	\end{equation}
	The numbers $ \sigma_{1}, ..., \sigma_{n} $ are called \textit{singular values} of $ A $. If $ \sigma_{r} > 0 $ is the smallest nonzero singular value, then the matrix $ A $ has rank $ r $.
\end{theorem}

The singular value decomposition exists for any matrix, we use $ m \geq n $ for simplicity. Note that the decomposition is not unique---only the matrix $ \Sigma $ is in general uniquely determined by $ A $. In the following, we will adopt the notation $ U = ( U_{1} ~ U_{2} ) $ and $ V = ( V_{1} ~ V_{2} ) $ with $ U_{1} \in \mathbb{R}^{m \times r} $, $ U_{2} \in \mathbb{R}^{m \times m-r} $, $ V_{1} \in \mathbb{R}^{n \times r} $ and $ V_{2} \in \mathbb{R}^{n \times n-r} $. We further denote $ \Sigma_{r} = \text{diag}(\sigma_{1}, ..., \sigma_{r}) \in \mathbb{R}^{r \times r} $. Then one can write
\begin{equation}\label{key}
A = \begin{pmatrix} U_{1} & U_{2} \end{pmatrix} \begin{pmatrix}
\Sigma_{r} & 0 \\
0 & 0
\end{pmatrix} \begin{pmatrix}
V_{1}^{T} \\
V_{2}^{T}
\end{pmatrix} = U_{1} \Sigma_{r} V_{1}^{T}
\end{equation}
This form will be called \textit{narrowed singular value decomposition}. It is useful for defining the pseudoinverse:

\begin{definition}
	Let $ A \in \mathbb{R}^{m \times n} $ be an $ m \times n $ matrix and $ A = U \Sigma V^{T} = U_{1} \Sigma_{r} V_{1}^{T} $ its (narrowed) singular value decomposition. Then the matrix $ A^{+} = V \Sigma^{+} U^{T} = V_{1} \Sigma_{r}^{+} U_{1}^{T} $ with $ \Sigma^{+} = \begin{pmatrix}
	\Sigma_{r}^{+} & 0 \\
	0 & 0
	\end{pmatrix} \in \mathbb{R}^{n \times m} $ and $ \Sigma_{r}^{+} = \text{diag} ( \sigma_{1}^{-1}, ..., \sigma_{r}^{-1}) \in \mathbb{R}^{r \times r} $ is called the \textit{Moore-Penrose pseudoinverse} of $ A $.
\end{definition}

\vspace{\baselineskip}

In the next we remind the fundamental spaces associated to the matrix $ A $.

\begin{definition}\label{rowcolnull}
	We define the following fundamental spaces:
	\begin{enumerate}[{(1)}]
		\item $ \mathcal{R}(A) = \{ y ~ \vert ~ \exists x \in \mathbb{R}^{n} : y = A x \} \subset \mathbb{R}^{m}  $ is the \textit{range} or \textit{column space}.
		\item $ \mathcal{R}(A^{T}) = \{ z ~ \vert ~ \exists y \in \mathbb{R}^{n} : z = A^{T} y \} \subset \mathbb{R}^{n}  $ is the \textit{row space}.
		\item $ \mathcal{N}(A) = \{ x ~ \vert ~ A x = 0 \} \subset \mathbb{R}^{n}  $ is the \textit{null space}.
	\end{enumerate}
\end{definition}

Note that since $ \mathcal{R}(A)^{\perp} = \mathcal{N}(A^{T}) $, it holds that $ \mathbb{R}^{m} = \mathcal{R}(A) \oplus \mathcal{N}(A^{T}) $. Analogically, one has $ \mathcal{R}(A^{T})^{\perp} = \mathcal{N}(A) $, therefore $ \mathbb{R}^{n} = \mathcal{R}(A^{T}) \oplus \mathcal{N}(A) $. With the help of the narrowed singular value decomposition $ A = U_{1} \Sigma_{r} V_{1}^{T} $ and the Moore-Penrose pseudoinverse, one can easily write down the projectors to these spaces:

\begin{theorem}
	The projectors to the spaces of Definition \ref{rowcolnull} are given by
	\begin{equation}\label{key}
	\begin{aligned}
	& P_{\mathcal{R}(A)} = AA^{+} = U_{1} U_{1}^{T} & \qquad & P_{\mathcal{R}(A^{T})} = A^{+} A = V_{1} V_{1}^{T} \\
	& P_{\mathcal{N}(A^{T})} = \mathbf{1} - AA^{+} = U_{2} U_{2}^{T} & \qquad & P_{\mathcal{N}(A)} = \mathbf{1} - A^{+} A = V_{2} V_{2}^{T}
	\end{aligned}
	\end{equation}
\end{theorem}

\vspace{\baselineskip}

The above described tools can be readily used to write an explicit solution to a general linear set of equations. Consider the matrix problem
\begin{equation}\label{Axb}
A x = b
\end{equation}
with $ A \in \mathbb{R}^{m \times n} $ a matrix, $ x \in \mathbb{R}^{n} $ and $ b \in \mathbb{R}^{m} $. Obviously, the equation only has a solution for $ x $ if $ b \in \mathcal{R}(A) $. This condition (or \textit{constraint}) can be equivalently stated as
\begin{equation}\label{Axbconstr}
U_{2} U_{2}^{T} b = 0
\end{equation}
where we assumed $ A = U_{1} \Sigma_{r} V_{1}^{T} $ and projected the right hand side of \eqref{Axb} onto $ \mathcal{N}(A^{T}) $ with $ P_{\mathcal{N}(A^{T})} = U_{2} U_{2}^{T} $. If the constraint \eqref{Axbconstr} holds, there is a family of solutions for $ x $ of the form
\begin{equation}\label{key}
x = A^{+} b + V_{2} c
\end{equation}
where $ c \in \mathbb{R}^{s} $ is an arbitrary vector of dimension $ s \equiv n - r $. By the way, this is the solution of the linear least squares problem $ A x \approx b $ which comes around by projecting the right-hand side of \eqref{Axb} onto $ \mathcal{R}(A) $ and thus solving the equation $ A x = A A^{+} b $ rather than \eqref{Axb}. Considering the linear least squares problem is equivalent to simply ignoring the constraint \eqref{Axbconstr}.\\

\vspace{\baselineskip}

In the rest, we shall very briefly recall the definition of a symplectic matrix and review its elementary properties. For more information, we refer to \cite{Gosson2006}.\\

\begin{definition}
	A \textit{symplectic matrix} $ W $ is a real $ 2q \times 2q $ matrix satisfying
	\begin{equation}\label{key}
	W^{T} \sigma W = \sigma
	\end{equation}
\end{definition}
with
\begin{equation}\label{key}
\sigma = \begin{pmatrix}
\phantom{-} 0 & \mathbf{1} \\
-\mathbf{1} & 0
\end{pmatrix}
\end{equation}

\begin{theorem}\label{sympl}
	Let us denote
	\begin{equation}\label{key}
	W = \begin{pmatrix}
	E & F \\
	G & H
	\end{pmatrix}
	\end{equation}
	where $ E, F, G, H $ are real $ q \times q $ matrices. Then the following conditions are equivalent:
	\begin{enumerate}
		\item The matrix $ W $ is symplectic.
		\item $ E^{T}G $, $ F^{T} H $ are symmetric and $ E^{T} H - G^{T} F = \mathbf{1} $
		\item $ E F^{T} $, $ G H^{T} $ are symmetric and $ E H^{T} - F G^{T} = \mathbf{1} $
	\end{enumerate}
\end{theorem}

\vspace{\baselineskip}

It follows from condition 2. that the inverse of a symplectic matrix $ W $ is
\begin{equation}\label{W-1}
W^{-1} = \begin{pmatrix}
\phantom{-}H^{T} & -F^{T} \\
-G^{T} & \phantom{-}E^{T}
\end{pmatrix}
\end{equation}

\vspace{\baselineskip}

\section{Discrete Linear Evolution}
We shall now apply the discrete canonical evolution formalism to the case of a linear dynamical system, using a matrix formulation of the problem. All the distinctive features of the (irregular) discrete evolution then become very clear, translated into the properties of the corresponding matrices.\\

We will consider the set $ \{0, ..., t\} \ni n $ of finite number of time-slices. We require that the maximum number $ q $ of degrees of freedom at any time-slice is finite, therefore, the system in question has a finite number $ N $ of degrees of freedom bounded by $ q(t+1) $. We define a \textit{solution} to be a point in the total phase space $ \mathcal{P}_{0t} = \bigotimes_{n = 0}^{t} \mathcal{P}_{n} $ satisfying the evolution equations; in particular, it must satisfy all the constraints. The set of solutions shall be denoted by $ \mathcal{S}_{0t} $. To obtain the time-slice data of a solution, one may use the projection $ \mathbb{P}_{n} : \mathcal{P}_{0t} \to \mathcal{P}_{n} $, given naturally as $ \mathbb{P}_{n} y = y_{n} $. Often the projection will be implicit. It may of course happen that the constraints are so severe that no solution on the whole interval $ \{0, ..., t\} $ exists; we are however more interested in the case when the solutions are plentiful.\\

Let us give a short remark on the action of the discussed system. In accordance with \eqref{action}, we are assuming that it is additive so that $ S(x) = \sum_{n = 0}^{t-1} S_{n+1}(x_{n}, x_{n+1}) $ where $ S_{n+1}(x_{n}, x_{n+1}) $ is the action contribution corresponding to the time-step between $ n $ and $ n+1 $. We know that the variation of $ S_{n+1}(x_{n}, x_{n+1}) $ produces the equations of motion in the form of momentum-matching \eqref{mommatch}. However, if we do not limit ourselves to solutions, we may as well consider any coordinate configuration and rewrite the action in terms of coordinates, pre-momenta and post-momenta. The single time-step action contribution then amounts to (in matrix notation)
\begin{equation}\label{Sn+1yy}
S_{n+1} = \textstyle \frac{1}{2} \left( x_{n+1}^{T} ~ ^{+}p_{n+1} - x_{n}^{T} ~ ^{-}p_{n} \right)
\end{equation}
The validity of this formula will become clear in a moment. In the terms used, the overall action between time-slices 0 and $ t $ is
\begin{equation}\label{Sgen}
S = \textstyle \frac{1}{2} \left( x_{t}^{T} ~ ^{+}p_{t} - x_{0}^{T} ~ ^{-}p_{0} \right) + \frac{1}{2} \sum_{n = 1}^{t-1} x_{n}^{T} \left( ^{+}p_{n} - ~ ^{-}p_{n} \right)
\end{equation}
For \textit{solutions} $ y \in \mathcal{S}_{0t} $, we can rewrite the action as a function on the solution space $ \mathcal{S} $,
\begin{equation}\label{Sy}
S(y) = \textstyle \frac{1}{2} \left( x_{t}^{T} p_{t} - x_{0}^{T} p_{0} \right)
\end{equation}
which, upon plugging in for momenta, can be seen as a discrete analogue of the Hamilton's principal function. Note that the Hamilton's principal function is given attention in Sections 7 and 9 of \cite{Hoehn2014}.\\

\vspace{\baselineskip}

Eventually let us define the symplectic structure $ \omega_{n} $ on $ \mathcal{P}_{n} \times \mathcal{P}_{n} $ by the usual prescription
\begin{equation}\label{symplstr}
\omega_{n}(y_{n},z_{n}) = x_{n}(y)^{T} p_{n}(z) - x_{n}(z)^{T} p_{n}(y)
\end{equation}
i.e., we confirm that $ (x_{n}, p_{n}) $ are canonical coordinates on $ \mathcal{P}_{n} $.\\

\vspace{\baselineskip}

\subsection{Matrix Formulation}

For a linear dynamical system, the equations (\ref{prepostmoneta}) can be written in matrix form
\begin{equation}\label{prepostmomentamatrix}
\begin{aligned}
^{-}p_{n} &= - \frac{\partial S_{n+1}}{\partial x_{n}} = L_{n} x_{n} + R_{n} x_{n+1} \\
^{+}p_{n+1} &= \frac{\partial S_{n+1}}{\partial x_{n+1}} =  \bar{L}_{n+1} x_{n} + \bar{R}_{n+1} x_{n+1}
\end{aligned}
\end{equation}
where $ ^{-}p_{n} $, $ ^{+}p_{n+1} $, $ x_{n} $, $ x_{n+1} $ are coordinate vectors of dimension $ q $ and $ L_{n}, R_{n}, \bar{L}_{n+1}, \bar{R}_{n+1} $ are $ q \times q $ matrices. If we consider a solution to the equations of motion, we may drop the $ - $ and $ + $ indices of the momenta, enforcing momentum matching. We remark that \eqref{prepostmomentamatrix} are identical to the equations (3.3) in \cite{Hoehn2014} although we employ a slightly different notation for the matrices.\\

In connection to \eqref{prepostmomentamatrix}, one should recall a well known fact that for linear systems, the second partial derivatives of action are symmetric; this is merely a consequence of the Young's (or Clairaut's, or Schwarz's) Theorem \cite{Young1908}. Taking this into account, one asserts that $ \frac{\partial S_{n+1}}{\partial x_{n A} \partial x_{n+1 B}} = \frac{\partial S_{n+1}}{\partial x_{n+1 B} \partial x_{n A}} $. Computing $ \frac{\partial}{\partial x_{n A}} \frac{\partial S_{n+1}}{\partial x_{n+1 B}} = \bar{L}_{n+1 BA} $ and $ \frac{\partial}{\partial x_{n+1 B}} \frac{\partial S_{n+1}}{\partial x_{n A}} = R_{n AB} $, one finds
\begin{equation}\label{pardevsymm}
\bar{L}_{n+1} = - R_{n}^{T}
\end{equation}
The similar assertions $ \frac{\partial S_{n+1}}{\partial x_{n A} \partial x_{n B}} = \frac{\partial S_{n+1}}{\partial x_{n B} \partial x_{n A}} $, $ \frac{\partial S_{n+1}}{\partial x_{n+1 A} \partial x_{n+1 B}} = \frac{\partial S_{n+1}}{\partial x_{n+1 B} \partial x_{n+1 A}} $ imply
\begin{equation}\label{pardevsymm2}
L_{n} = L_{n}^{T}, \qquad \bar{R}_{n+1} = \bar{R}_{n+1}^{T}
\end{equation}
i.e., the matrices $ L_{n} $ and $ \bar{R}_{n+1} $ are \textit{symmetric}. This will make our later work substantially easier.\\

One may deduce from \eqref{prepostmomentamatrix} the general form of the action contribution $ S_{n+1}(x_{n}, x_{n+1}) $,
\begin{equation}\label{Slag}
\begin{aligned}
S_{n+1} &= \textstyle \frac{1}{2} \left( x_{n+1}^{T} \bar{L}_{n+1} x_{n} + x_{n+1}^{T} \bar{R}_{n+1} x_{n+1} - x_{n}^{T} L_{n} x_{n} - x_{n}^{T} R_{n} x_{n+1} \right) =\\
&= \textstyle - \frac{1}{2} \left(  x_{n}^{T} L_{n} x_{n} + 2 x_{n}^{T} R_{n} x_{n+1} - x_{n+1}^{T} \bar{R}_{n+1} x_{n+1} \right)
\end{aligned}
\end{equation}
where we have used the identity \eqref{pardevsymm} to simplify the expression. Turning once again to \eqref{prepostmomentamatrix}, one can easily express this quantity in mixed variables as \eqref{Sn+1yy}. Let us stress out that although quite convenient, the form \eqref{Sn+1yy} is not relevant form the perspective of Hamiltonian mechanics, since it is not given in terms of canonical coordinates. Note that we have not provided any Hamiltonian function at all. Therefore, the best we can do is to employ \eqref{Slag} which can be viewed as the Lagrangian and express the Hamiltonian evolution map directly from it. This is what we are going to do in the next paragraph. Meanwhile, let us point out that \eqref{Slag} has its equivalent in eq. (3.1) of \cite{Hoehn2014}. The two expressions for action are related via $ L_{n} \equiv - a^{n+1} $, $ R_{n} \equiv - c^{n+1} $ and $ \bar{R}_{n+1} \equiv b^{n} $.

\vspace{\baselineskip}
\vspace{\baselineskip}

We will now give an explicit formulation of the (forward) Hamiltonian evolution in terms of the matrices introduced in \eqref{prepostmomentamatrix}. We start by solving the pre-momentum equation
\begin{equation}
R_{n} x_{n+1} = p_{n} - L_{n} x_{n}
\end{equation}
The matrix $ R_{n} $ is generally singular. Upon using the narrowed singular value decomposition, we find that the solution exists when
\begin{equation}\label{xn+1preconstraint}
P_{\mathcal{N}(R_{n}^{T})} \left( p_{n} - L_{n} x_{n} \right) = 0
\end{equation}
with the projector to the null space of $ R_{n}^{T} $ given by
\begin{equation}\label{key}
P_{\mathcal{N}(R_{n}^{T})} = U_{2}(R_{n}) U_{2}(R_{n})^{T}
\end{equation}
Next, if the solution does exist, it is of the form
\begin{equation}\label{xn+1solution}
x_{n+1} = R_{n}^{+} \left( p_{n} - L_{n} x_{n} \right) + V_{2}(R_{n}) \lambda_{n+1}
\end{equation}
where
\begin{equation}
R_{n}^{+} = V_{1}(R_{n}) ~ \Sigma_{r}(R_{n})^{-1} ~ U_{1}(R_{n})^{T}
\end{equation}
is the Moore-Penrose pseudoinverse of $ R_{n} $ and $ \lambda_{n+1} $ is an arbitrary vector of dimension $ s_{n} \equiv q - r_{n} $ with $ r_{n} \equiv \text{rank}(R_{n}) $. Note that the equation \eqref{xn+1preconstraint} is a matrix form of the pre-constraint, defining the linear pre-constraint subspace $ \mathcal{C}_{n}^{-} \subset \mathcal{P}_{n} $. We can rewrite it in a more compact form
\begin{equation}\label{preconst}
C_{n} y_n = 0, \qquad C_{n} = \begin{pmatrix}
- P_{\mathcal{N}(R_{n}^{T})} L_{n} & & P_{\mathcal{N}(R_{n}^{T})}
\end{pmatrix}
\end{equation}
where
\begin{equation}\label{key}
y_n = \begin{pmatrix}
x_{n} \\
p_{n}
\end{pmatrix}
\end{equation}
is a coordinate vector of dimension $ 2q $ representing a point $ y_{n} \in \mathcal{P}_{n} $ and the two-block matrix $ C_{n} $ has dimensions $ q \times 2q $. Upon plugging the solution \eqref{xn+1solution} into the post-momenta equation in \eqref{prepostmomentamatrix}, one finds
\begin{equation}\label{evol}
y_{n+1} = E_{n} y_{n} + F_{n+1} \lambda_{n+1}
\end{equation}
with
\begin{equation}\label{En}
E_{n} = \begin{pmatrix}
- R_{n}^{+} L_{n} & R_{n}^{+} \\
\bar{L}_{n+1} - \bar{R}_{n+1} R_{n}^{+} L_{n} & \bar{R}_{n+1} R_{n}^{+}
\end{pmatrix} = \begin{pmatrix}
- R_{n}^{+} L_{n} & R_{n}^{+} \\
- R_{n}^{T} - \bar{R}_{n+1} R_{n}^{+} L_{n} & \bar{R}_{n+1} R_{n}^{+}
\end{pmatrix}
\end{equation}
and
\begin{equation}\label{Fn+1}
F_{n+1} = \begin{pmatrix}
V_{2}(R_{n}) \\
\bar{R}_{n+1} V_{2}(R_{n})
\end{pmatrix}
\end{equation}
The matrix $ E_{n} $ has dimensions $ 2q \times 2q $ and the matrix $ F_{n+1} $ has dimensions $ 2q \times s_{n} $. The coordinate vectors $ y_{n+1} $ are elements of the post-constraint surface $ \mathcal{C}_{n+1}^{+} \subset \mathcal{P}_{n+1} $. The matrix formulation has made explicit the dependence of the evolution on a free parameter $ \lambda_{n+1} \in \mathbb{R}^{s_{n}} $. We shall acknowledge this by denoting the \textit{Hamiltonian evolution map}, originally introduced in \eqref{H}, by $ \mathbb{H}_{n+1}(\lambda_{n+1}) : \mathcal{C}_{n}^{-} \rightarrow \mathcal{C}_{n+1}^{+} $. Viewed in this way, the map is unique, but \textit{not onto}, because the parameter $ \lambda_{n+1} $ picks a particular subspace $ \mathcal{C}_{n+1}^{+(\lambda_{n+1})} \equiv \mathbb{H}_{n+1}(\lambda_{n+1}) \mathcal{C}_{n}^{-} \subset \mathcal{C}_{n+1}^{+} $ as the image. Thanks to the linearity of the evolution equations, the post-constraint surface $ \mathcal{C}_{n+1}^{+} $ can be viewed as an affine space over $ \mathcal{C}_{n+1}^{+(0)} \equiv \mathbb{H}_{n+1}(0) \mathcal{C}_{n}^{-} $ with the point-set $ \Lambda_{n+1} \equiv \{ F_{n+1} \lambda_{n+1} ~ \vert ~ \lambda_{n+1}  \in \mathbb{R}^{s_{n}} \} $. It is very much correct to treat it this way, distinguishing between \textit{vectors} $ E_{n} y_{n} \in \mathcal{C}_{n+1}^{+(0)} $ and \textit{points} $ F_{n+1} \lambda_{n+1} \in \Lambda_{n+1} $. Note that both $ \mathcal{C}_{n+1}^{+(0)} $ and $ \Lambda_{n+1} $ are linear subspaces of $ \mathcal{P}_{n+1}  $. However, one can say more:

\begin{observation}
	The subspaces $ \mathcal{C}_{n+1}^{+(0)} $ and $ \Lambda_{n+1} $ have zero intersection.
\end{observation}
\begin{proof}
	It suffices to look at the coordinate parts of $ E_{n} y_{n} $ and $ F_{n+1} \lambda_{n+1} $. These are given simply by the two terms in \eqref{xn+1solution}. Since $ R_{n}^{+} = V_{1}(R_{n}) ~ \Sigma_{r}(R_{n})^{-1} ~ U_{1}(R_{n})^{T} $ and the columns of $ V_{1}(R_{n}) $ are by definition orthogonal to the columns of $ V_{2}(R_{n}) $, it follows that the two parts are orthogonal. The assertion is implied.
\end{proof}

\vspace{\baselineskip}

As a consequence of the above observation, the post-constraint surface $ \mathcal{C}_{n+1}^{+} $ can be viewed as a linear subspace $ \mathcal{C}_{n+1}^{+} = \mathcal{C}_{n+1}^{+(0)} \oplus \Lambda_{n+1} $ of $ \mathcal{P}_{n+1} $, rather than an affine space. This further simplifies our work.\\

\vspace{\baselineskip}

For completeness, we provide also the explicit form of the backward Hamiltonian evolution. It is fully analogical to the preceding case. To express the evolution from time-step $ n+1 $ to time-step $ n $, one has to solve the post-momentum equation in \eqref{prepostmomentamatrix}, which yields a constraint
\begin{equation}\label{backconst}
P_{\mathcal{N}(\bar{L}_{n+1}^{T})} \left( p_{n+1} - \bar{R}_{n+1} x_{n+1} \right) = 0
\end{equation}
with the projector
\begin{equation}\label{key}
P_{\mathcal{N}(\bar{L}_{n+1}^{T})} = U_{2}(\bar{L}_{n+1}) U_{2}(\bar{L}_{n+1})^{T}
\end{equation}
and a solution
\begin{equation}\label{key}
x_{n} = \bar{L}_{n+1}^{+} \left( p_{n+1} - \bar{R}_{n+1} x_{n+1} \right) + V_{2}(\bar{L}_{n+1}) \mu_{n}
\end{equation}
with an arbitrary vector $ \mu_{n} \in \mathbb{R}^{\bar{s}_{n}} $ where we denote $ \bar{s}_{n} = q - \bar{r}_{n} $ with $ \bar{r}_{n} = \text{rank}(\bar{L}_{n+1}) $. The constraint \eqref{backconst} can be rewritten compactly as
\begin{equation}\label{postconst}
\bar{C}_{n+1} y_{n+1} = 0, \qquad \bar{C}_{n+1} = \begin{pmatrix}
- P_{\mathcal{N}(\bar{L}_{n+1}^{T})} \bar{R}_{n+1} & P_{\mathcal{N}(\bar{L}_{n+1}^{T})}
\end{pmatrix}
\end{equation}
When it is satisfied, a solution to the backwards evolution exists and is given by
\begin{equation}
y_{n} = \bar{E}_{n+1} y_{n+1} + \bar{F}_{n} \mu_{n}
\end{equation}
with
\begin{equation}
\bar{E}_{n+1} = \begin{pmatrix}
- \bar{L}_{n+1}^{+} \bar{R}_{n+1} & \bar{L}_{n+1}^{+} \\
R_{n} - L_{n} \bar{L}_{n+1}^{+} \bar{R}_{n+1} & L_{n} \bar{L}_{n+1}^{+}
\end{pmatrix}
\end{equation}
and
\begin{equation}
\bar{F}_{n} = \begin{pmatrix}
V_{2}(\bar{L}_{n+1}) \\
L_{n} V_{2}(\bar{L}_{n+1})
\end{pmatrix}
\end{equation}

\vspace{\baselineskip}
\vspace{\baselineskip}

\subsection{The Constraint Surfaces and the Symplectic Structure}\label{constrsurf}
We devote this subsection to the study of constraint surfaces and the general properties of the Hamiltonian evolution map $ \mathbb{H}_{n+1}(\lambda_{n+1}) $ with respect to the symplectic structures $ \omega_{n} $ and $ \omega_{n+1} $ given by \eqref{symplstr}. We start with Theorem \ref{comegac}, which is nothing but a variation on the equation \eqref{dots0}, saying that ``symplectic product of solutions is conserved''. In our context, Theorem \ref{comegac} can be given in this exact simple wording. Assume that $ y_{n}, z_{n} \in \mathcal{C}_{n}^{-} $. From \eqref{prepostmomentamatrix}, one gets
\begin{equation}\label{omeganA}
\begin{aligned}
\omega_{n}(y_{n},z_{n}) &= \phantom{-} x_{n}(y)^{T} L_{n} x_{n}(z) + x_{n}(y)^{T} R_{n} x_{n+1}(z) \\ & \phantom{=} - x_{n}(z)^{T} L_{n} x_{n}(y) - x_{n}(z)^{T}  R_{n} x_{n+1}(y) = \\
&= x_{n}(y)^{T} R_{n} x_{n+1}(z) - x_{n}(z)^{T}  R_{n} x_{n+1}(y)
\end{aligned}
\end{equation}
where we used the symmetry of $ L_{n} $. Forward Hamiltonian evolution with initial conditions $ y_{n}, z_{n} $ then yields
\begin{equation}
\begin{aligned}
\omega_{n+1}(y_{n+1},z_{n+1}) &= \phantom{-} x_{n+1}(y)^{T} \bar{L}_{n+1} x_{n}(z) + x_{n+1}(y)^{T} \bar{R}_{n+1} x_{n+1}(z) \\ & \phantom{=} - x_{n+1}(z)^{T} \bar{L}_{n+1} x_{n}(y) - x_{n+1}(z)^{T} \bar{R}_{n+1} x_{n+1}(y) = \\
&= \phantom{-} x_{n+1}(y)^{T} \bar{L}_{n+1} x_{n}(z) - x_{n+1}(z)^{T} \bar{L}_{n+1} x_{n}(y) = \\
&= - x_{n}(z)^{T} R_{n} x_{n+1}(y) + x_{n}(y)^{T} R_{n} x_{n+1}(z) = \\
&= \omega_{n}(y_{n},z_{n})
\end{aligned}
\end{equation}
where we used symmetry of $ \bar{R}_{n+1} $, exploited the relation \eqref{pardevsymm} and compared with \eqref{omeganA}. Similarly for backwards evolution: if $ y_{n}, z_{n} \in \mathcal{C}_{n}^{+} $, then
\begin{equation}\label{omeganB}
\omega_{n}(y_{n},z_{n}) = x_{n}(y)^{T} \bar{L}_{n} x_{n-1}(z) - x_{n}(z)^{T} \bar{L}_{n} x_{n-1}(y)
\end{equation}
and we find
\begin{equation}
\begin{aligned}
\omega_{n-1}(y_{n-1},z_{n-1}) &= \phantom{-} x_{n-1}(y)^{T} R_{n-1} x_{n}(z) - x_{n-1}(z)^{T} R_{n-1} x_{n}(y) =\\
&= - x_{n}(z)^{T} \bar{L}_{n} x_{n-1}(y) + x_{n}(y)^{T} \bar{L}_{n} x_{n-1}(z) =\\
&= \omega_{n}(y_{n},z_{n})
\end{aligned}
\end{equation}

We have thus explicitly checked that the symplectic product of solutions is indeed conserved, understanding that a solution $ y_{n} $ of the equations of motion satisfies momentum-matching as well as \textit{all the constraints}. We will later extend our discussion of solutions to the case of multiple time-steps. At this point we can move on to the analysis of the constraint surfaces.\\

It is essential to understand that while $ \omega_{n} $ is symplectic on $ \mathcal{P}_{n} $, it is generally not symplectic on $ \mathcal{C}_{n}^{-} \subset \mathcal{P}_{n} $ because there it may be degenerate. In other words, the constraint surface $ \mathcal{C}_{n}^{-} $ is in general not a symplectic subspace of $ \mathcal{P}_{n} $. This is unfortunate for many applications which require to have a well defined evolution map between symplectic spaces, as is the case of canonical quantization. Luckily, we can easily reduce the constraint spaces by a standard procedure to make them symplectic and establish such an evolution map. We do this below, first formally and then explicitly.\\

Let us describe the procedure in general terms for an arbitrary linear subspace $ \mathcal{C} $ of the symplectic space $ (\mathcal{P}, \omega) $. We first define the corresponding \textit{null space} $ \mathcal{N}_{\omega}(\mathcal{C}) $ as
\begin{equation}\label{key}
\mathcal{N}_{\omega}(\mathcal{C}) = \{ z \in \mathcal{C} ~ \vert ~ \omega(z, u) = 0 ~ \forall u \in \mathcal{C} \}
\end{equation}
The reader may notice that $ \mathcal{N}_{\omega}(\mathcal{C}) = \mathcal{C} \cap \mathcal{C}^{\omega} $ where $ \mathcal{C}^{\omega} = \{ z \in \mathcal{P} ~ \vert ~ \omega(z, u) = 0 ~ \forall u \in \mathcal{C} \} $ is so-called \textit{skew-orthogonal set} to $ \mathcal{C} $, see Sec. 1.2 in \cite{Gosson2006} for details. For us, $ \mathcal{N}_{\omega}(\mathcal{C}) $ is simply the subspace of $ \mathcal{C} $ which makes $ \omega $ degenerate. If $ \mathcal{N}_{\omega}(\mathcal{C}) = \{ 0 \} $, then $ \mathcal{C} $ is symplectic and we are done. Otherwise we proceed with a second step, in which we get rid of the degeneracy: we take $ \tilde{\mathcal{C}} = \mathcal{C} / \mathcal{N}_{\omega}(\mathcal{C}) $. This is the space of equivalence classes $ [ y ] = \{ y + z ~ \vert ~ z \in \mathcal{N}_{\omega}(\mathcal{C}) \} $ of equivalent $ y \in \mathcal{C} $. This means that if $ y, \tilde{y} \in \mathcal{C} $ are such that $ \tilde{y} = y + z $ with $ z \in \mathcal{N}_{\omega}(\mathcal{C}) $, then $ [y] = [\tilde{y}] $. One can also put it differently:
\begin{definition}
	We say that $ y, \tilde{y} \in \mathcal{C} $ are \textit{symplectically equivalent} on $ \mathcal{C} $ w.r.t. $ \omega $ and write $ y \stackrel{\mathcal{C}}{\sim} \tilde{y} $, if $ \omega(y,u) = \omega(\tilde{y},u) $ for all $ u \in \mathcal{C} $.
\end{definition}

\begin{observation}\label{Osim}
	The space $ \tilde{\mathcal{C}} $ is composed of equivalence classes $ [ y ] = \{ \tilde{y} \in \mathcal{C} \vert y \stackrel{\mathcal{C}}{\sim} \tilde{y} \} $ of symplectically equivalent vectors in $ \mathcal{C} $.
\end{observation}
\begin{proof}
	A simple exercise.
\end{proof}

The point of this procedure is captured by the following observation.
\begin{observation}\label{Octilde}
	Let $ \omega : \tilde{\mathcal{C}} \times \tilde{\mathcal{C}} \rightarrow \mathbb{R} $ be defined by $ \omega([y],[u]) = \omega(y,u) $ where $ y \in [y] $, $ u \in [u] $ are arbitrarily chosen representatives of their classes. Then $ (\tilde{\mathcal{C}}, \omega) $ is symplectic.
\end{observation}
\begin{proof}
	First of all, note that according to Observation \ref{Osim}, the definition of $ \omega : \tilde{\mathcal{C}} \times \tilde{\mathcal{C}} \rightarrow \mathbb{R} $ is consistent. Next, by assigning $ [y] + c [u] = [y + c u] $ (for $ y, u \in \mathcal{C} $ and $ c \in \mathbb{R} $), we let $ \tilde{\mathcal{C}} $ inherit the natural vector-space structure of $ \mathcal{C} $. Suppose that $ \omega( [y], [z] ) = 0 $ for all $ [y] \in \tilde{\mathcal{C}} $. Then it holds $ \omega(y,z) = 0 $ for all $ y \in \mathcal{C} $. In other words, $ \omega(y,z) = \omega(y,0) $ for all $ y \in \mathcal{C} $, i.e., $ z \stackrel{\mathcal{C}}{\sim} 0 $ and $ [z] = [0] $. It follows that $ \omega $ on $ \tilde{\mathcal{C}} $ is non-degenerate. Since it is also bilinear and antisymmetric, it is symplectic.
\end{proof}

Because equivalence classes are not very practical for computations, we add an optional third step, which is to consider a representative space $ \dot{\mathcal{C}} \subset \mathcal{C} $ such that each $ y \in \dot{\mathcal{C}} $ corresponds to a class $ [y] \in \tilde{\mathcal{C}} $. In other words, $ \dot{\mathcal{C}} $ is a set of symplectically inequivalent vectors $ y \in \mathcal{C} $ originating by picking one particular element from each class. We strongly prefer a choice of $ \dot{\mathcal{C}} $ which is a \textit{linear subspace} of $ \mathcal{C} $ so that the vector-space operations on $ \dot{\mathcal{C}} $ align with those on $ \mathcal{C} $; we shall therefore assume this. A possible way to find such $ \dot{\mathcal{C}} $ is to pick a maximal linearly independent set of symplectically inequivalent vectors in $ \mathcal{C} $ and generating $ \dot{\mathcal{C}} $ as linear span of this set. Needless to say, $ \omega $ naturally carries over from $ \mathcal{C} $ to $ \dot{\mathcal{C}} $, making $ (\dot{\mathcal{C}}, \omega) $ symplectic. One can make the following observation.
\begin{observation}\label{cdotoplusn}
	It holds $ \mathcal{C} = \dot{\mathcal{C}} \oplus  \mathcal{N}_{\omega}(\mathcal{C}) $.
\end{observation}
\begin{proof}
	Any $ \tilde{y} \in \mathcal{C} $ gives rise to a unique $ [ \tilde{y} ] \in \tilde{\mathcal{C}} $, which is in turn uniquely represented by $ y \in \dot{\mathcal{C}} $. We denote $ \tilde{y} - y \equiv z $, then by definition $ z \in \mathcal{N}_{\omega}(\mathcal{C}) $. It follows that there is a unique decomposition $ \tilde{y} = y + z $ of the vector $ \tilde{y} \in \mathcal{C} $ into $ y \in \dot{\mathcal{C}} $ and $ z \in \mathcal{N}_{\omega}(\mathcal{C}) $, hence the assertion.
\end{proof}

The above observation suggests a special choice of $ \dot{\mathcal{C}} $ which is uniquely fixed by an inner product on $ \mathcal{C} $. It is of course $ \dot{\mathcal{C}} = \mathcal{N}_{\omega}(\mathcal{C})^{\perp} $. We use this choice below on some occasions with the canonical inner product.\\

We will now use the described procedure for the pre-constraint surface $ \mathcal{C}_{n}^{-} $, resulting in the symplectic space $ \dot{\mathcal{C}}_{n}^{-} \subset \mathcal{P}_{n} $. Once we have it defined, we would like to see how $ \dot{\mathcal{C}}_{n}^{-} $ evolves to the next time-slice.\\

\vspace{\baselineskip}

\subsection{The Adapted Coordinates}

In what comes next, we would like to offer a particular choice of the representative spaces $ \dot{\mathcal{C}}_{n}^{-}, \dot{\mathcal{C}}_{n+1}^{+} $ and a construction of symplectic bases of $ \mathcal{P}_{n}, \mathcal{P}_{n+1} $ that will be adapted to this choice. Before we start, let us give a short comment on how our approach differs from that of \cite{Hoehn2014}. In Sec. 6 of that reference, a special choice of basis on the phase space is introduced. As it will be in our case, the associated linear transformation given by eq. (6.1) is canonical (so the new basis is symplectic). Its purpose is to separate the primary and secondary constraints. With a subsequent second transformation, it is then possible to trivialize the constraints, which opens the door to a classification of the propagating degrees of freedom constituted by observables and free parameters. Further analysis describes how these propagate from time-slice $ n-1 $ through time-slice $ n $ to time-slice $ n+1 $.\\

Our approach is quite different, and perhaps more lightweight. We focus on a single time-step from $ n $ to $ n+1 $ and find symplectic bases in both the involved phase spaces which disentangle the most important subspaces induced by the symplectic structure, as discussed in the previous section. In doing so, we also choose the representative spaces of symplectically equivalent time-slice data and build their symplectic bases. Incidentally (due to our rather natural choice), the corresponding symplectic (or to say, canonical) transformation will turn out to trivialize the one-step Hamiltonian evolution map.\\

First let us parametrize $ \mathcal{C}_{n}^{-} $. To do that, we need to solve the pre-constraint \eqref{preconst} or, equivalently, \eqref{xn+1preconstraint}. The latter can be solved trivially by writing
\begin{equation}\label{key}
p_{n} = L_{n} x_{n} + \theta_{n}
\end{equation}
with an arbitrary $ \theta_{n} \in \mathcal{R}(R_{n}) $ in the column space of $ R_{n} $; recall that $ \mathcal{R}(R_{n}) \perp \mathcal{N}(R_{n}^{T}) $. Then we have a general form of $ y_{n} \in \mathcal{C}_{n}^{-} $,
\begin{equation}\label{key}
y_{n} = \begin{pmatrix}
x_{n} \\
L_{n} x_{n} + \theta_{n}
\end{pmatrix}
\end{equation}
and we see that $ \dim \mathcal{C}_{n}^{-} = q + r_{n} $ with $ r_{n} = \dim\mathcal{R}(R_{n}) = \text{rank}(R_{n}) $. Now we take a moment to parametrize $ \mathcal{N}_{\omega}(\mathcal{C}_{n}^{-}) $. Recall that $ z_{n} \in \mathcal{N}_{\omega}(\mathcal{C}_{n}^{-}) $ must be in $ \mathcal{C}_{n}^{-} $ and must satisfy $ \omega_{n}(y_{n}, z_{n}) = 0 $ for all $ y_{n} \in \mathcal{C}_{n}^{-} $. These two conditions give us
\begin{equation}\label{key}
z_{n} = \begin{pmatrix}
\mu_{n} \\
L_{n}^{T} \mu_{n}
\end{pmatrix}
\end{equation}
with $ \mu_{n} \in \mathcal{N}(R_{n}^{T}) $ such that $ (L_{n}^{T} - L_{n}) \mu_{n} \in \mathcal{R}(R_{n}) $. However, the second condition is satisfied automatically, thanks to the symmetry of $ L_{n} $. We then see that $ \dim \mathcal{N}_{\omega}(\mathcal{C}_{n}^{-}) = \dim \mathcal{N}(R_{n}^{T}) = s_{n} $ with $ s_{n} \equiv q - r_{n} $. Now we can already guess the possible structure of $ \dot{\mathcal{C}}_{n}^{-} $. Let us write the relation $ \tilde{y}_{n} = y_{n} + z_{n} $ in the parametrized form,
\begin{equation}\label{key}
\begin{pmatrix}
\tilde{x}_{n} \\
L_{n} \tilde{x}_{n} + \tilde{\theta}_{n}
\end{pmatrix} = \begin{pmatrix}
x_{n} \\
L_{n} x_{n} + \theta_{n}
\end{pmatrix} + \begin{pmatrix}
\mu_{n} \\
L_{n}^{T} \mu_{n}
\end{pmatrix} = \begin{pmatrix}
x_{n} + \mu_{n} \\
L_{n} \left( x_{n} + \mu_{n} \right) + \theta_{n}
\end{pmatrix}
\end{equation}
and we see that $ \tilde{x}_{n} = x_{n} + \mu_{n} $ and $ \tilde{\theta}_{n} = \theta_{n} $. Therefore it is natural to pick a unique $ y_{n} $ from each class $ [ y_{n} ] $ solemnly by fixing $ x_{n} = \varrho_{n} \in \mathcal{R}(R_{n}) $ orthogonal to $ \mu_{n} $. That is, we chose $ y_{n} \in \dot{\mathcal{C}}_{n}^{-} $ to be of the form
\begin{equation}\label{key}
y_{n} = \begin{pmatrix}
\varrho_{n} \\
L_{n} \varrho_{n} + \theta_{n}
\end{pmatrix}
\end{equation}
with $ \varrho_{n} \in \mathcal{R}(R_{n}) $ and still $ \theta_{n} \in \mathcal{R}(R_{n}) $. Clearly, this is a linear subspace of $ \mathcal{C}_{n}^{-} $, as we require. Also, $ \dim \dot{\mathcal{C}}_{n}^{-} = 2 r_{n} $ is indeed even, as expected from a symplectic space.\\

Thanks to the very simple structure of $ \tilde{\mathcal{C}}_{n}^{-} $ and subsequently of $ \dot{\mathcal{C}}_{n}^{-} $, it is possible to construct a natural symplectic basis for $ \dot{\mathcal{C}}_{n}^{-} $. We shall denote it by $ \{ \dot{e}_{nX} \} $ with $ X = 1,..., r_{n}, q+1, ..., q+r_{n} $. Knowing that the basis vectors have the structure
\begin{equation}\label{key}
\dot{e}_{nX} = \begin{pmatrix}
\varrho_{nX} \\
L_{n} \varrho_{nX} + \theta_{nX}
\end{pmatrix}
\end{equation}
(the indices do not represent components here, but give names to vectors) one expresses the requirement of symplecticity as
\begin{equation}\label{omegaedot}
\sigma_{XY} = \omega_{n}( \dot{e}_{nX}, \dot{e}_{nY} ) = \varrho_{nX}^{T} \theta_{nY} - \theta_{nX}^{T} \varrho_{nY}
\end{equation}

Now recall the singular value decomposition $ R_{n} = U(R_{n}) \Sigma(R_{n}) V(R_{n})^{T} $. The orthogonal matrices $ U(R_{n}), V(R_{n}) $ are not unique; one can use any suitable choice of these two. In the following, we shall cease to write arguments of all the matrices arising from the singular value decomposition of $ R_{n} $, i.e., we denote $ U \equiv U(R_{n}) $, $ V \equiv V(R_{n}) $, $ \Sigma \equiv \Sigma(R_{n}) $ and so on. We will occasionally use this condensed notation.\\

We shall satisfy the equation \eqref{omegaedot} by the deliberate choice $ \varrho_{nE} = 0 $, $ \varrho_{nE+q} = - U_{1} \varepsilon_{nE} $, $ \theta_{nE} = - \varrho_{nE+q} $, $ \theta_{nE+q} = 0 $ for all $ E = 1, ..., r_{n} $. Here we use the canonical basis $ \{ \varepsilon_{nE} \}_{E = 1}^{r_{n}} $ in $ \mathbb{R}^{r_{n}} $ to generate the set $ \{ U_{1} \varepsilon_{nE} \}_{E = 1}^{r_{n}} $ which is an orthonormal basis of $ \mathcal{R}(R_{n}) $. Recall that $ U_{1} $ is a $ q \times r_{n} $ submatrix of $ U = \begin{pmatrix} U_{1} & U_{2} \end{pmatrix} $. We end up with
\begin{equation}\label{doteE}
\dot{e}_{nE} = \begin{pmatrix}
0 \\
U_{1} \varepsilon_{nE} 
\end{pmatrix}, \qquad \dot{e}_{nE+q} = - \begin{pmatrix}
U_{1} \varepsilon_{nE} \\
L_{n} U_{1} \varepsilon_{nE} 
\end{pmatrix}
\end{equation}
Surely we could have chosen a different arrangement of the vectors, but this one is the most suitable for our cause. In any way, we have a symplectic basis of $ \dot{\mathcal{C}}_{n}^{-} $. Since $ \dot{\mathcal{C}}_{n}^{-} \subset \mathcal{P}_{n} $, we know by virtue of the ``symplectic Gram–Schmidt theorem'', see Sec. 1.2 of \cite{Gosson2006}, that it can be extended to a full symplectic basis of $ \mathcal{P}_{n} $. We shall do this by fixing
\begin{equation}\label{doteM}
\dot{e}_{nM} = \begin{pmatrix}
0 \\
U_{2} \iota_{nM}
\end{pmatrix}, \qquad \dot{e}_{nM+q} = - \begin{pmatrix}
U_{2} \iota_{nM} \\
L_{n} U_{2} \iota_{nM} 
\end{pmatrix}
\end{equation}
where $ \{ \iota_{nM} \}_{M = r_{n}+1}^{q} $ is the (oddly numbered) canonical basis of $ \mathbb{R}^{s_{n}} $ and obviously $ M = r_{n}+1, ..., q $. It is clear that the vectors $ \dot{e}_{nI} $ are all linearly independent. Thus we have fixed a basis $ \{ \dot{e}_{nI} \}_{I = 1}^{2q} $ of $ \mathcal{P}_{n} $. The reader can easily check, using the analogue of \eqref{omegaedot}, that it is indeed symplectic. We will verify this one step later in another way. Before that however, it is worthy noticing that $ \{ \dot{e}_{nM+q} \}_{M = r_{n}+1}^{q} $ is a basis of $ \mathcal{N}_{\omega}(\mathcal{C}_{n}^{-}) $. This is a reminder of the fact that $ \mathcal{C}_{n}^{-} = \dot{\mathcal{C}}_{n}^{-} \oplus  \mathcal{N}_{\omega}(\mathcal{C}_{n}^{-}) $. On the other hand, $ \{ \dot{e}_{nM} \}_{M = r_{n}+1}^{q} $ is a basis of $ \mathcal{C}_{n}^{- \perp} $ (the orthogonal complement in $ \mathcal{P}_{n} $ is taken w.r.t. the canonical inner product here). This is the space of vectors that do not satisfy the pre-constraint.\\

We formalize the basis transformation as
\begin{equation}\label{doteWe}
e_{nJ} = \dot{e}_{nI} \dot{W}_{nIJ}
\end{equation}
and read out from \eqref{doteE} and \eqref{doteM} the inverse matrix
\begin{equation}\label{dotWinv}
\dot{W}_{n}^{-1} = \begin{pmatrix}
0 & - U \\
U & - L_{n} U
\end{pmatrix}
\end{equation}
Now the promised verification: because $ U^{T} L_{n} U $ is symmetric and $ U^{T} U = \mathbf{1} $, it follows from Theorem \ref{sympl} that $ \dot{W}_{n}^{-1} $ is a \textit{symplectic matrix}. Its inverse is also symplectic, and using \eqref{W-1}, we find that it is
\begin{equation}\label{dotW}
\dot{W}_{n} = \begin{pmatrix}
-U^{T} L_{n} & U^{T} \\
-U^{T} & 0
\end{pmatrix}
\end{equation}
Thus we can take for granted that \eqref{doteWe} is a symplectic transformation passing from the new symplectic basis to the canonical one.\\

\vspace{\baselineskip}

We will shortly continue our discourse by evolving the described subspaces of $ \mathcal{P}_{n} $ to $ \mathcal{P}_{n+1} $ using the prescription \eqref{evol}. Before we do so, we want to remark that because of the conservation of the symplectic structure, one automatically obtains $ \mathbb{H}_{n+1}(\lambda_{n+1}) \mathcal{N}_{\omega}(\mathcal{C}_{n}^{-}) = \mathcal{N}_{\omega}(\mathbb{H}_{n+1}(\lambda_{n+1}) \mathcal{C}_{n}^{-}) = \mathcal{N}_{\omega}( \mathcal{C}_{n+1}^{+(\lambda_{n+1})} ) $ and consequently also $ \mathbb{H}_{n+1}(\lambda_{n+1}) \tilde{\mathcal{C}}_{n}^{-} = \tilde{\mathcal{C}}_{n+1}^{+(\lambda_{n+1})} $ for arbitrary $ \lambda_{n+1} $. It therefore seems natural to generate our representative space $ \dot{\mathcal{C}}_{n+1}^{+} $ in the post-constraint surface by evolving the representative space in the pre-constrain surface, i.e., fix $ \dot{\mathcal{C}}_{n+1}^{+} = \mathbb{H}_{n+1}(\lambda_{n+1}) \dot{\mathcal{C}}_{n}^{-} $ with some $ \lambda_{n+1} $. This is indeed possible. The symplectic structure of $ \dot{\mathcal{C}}_{n}^{-} $ will not be touched by $ \mathbb{H}_{n+1}(\lambda_{n+1}) $, and it is therefore guaranteed that $ \dot{\mathcal{C}}_{n+1}^{+} $ will be also a symplectic space of the same dimension. For simplicity, we shall opt for $ \dot{\mathcal{C}}_{n+1}^{+} = \mathbb{H}_{n+1}(0) \dot{\mathcal{C}}_{n}^{-} $.\\

Now let us proceed with the calculation. We start with $ z_{n} \in \mathcal{N}_{\omega}(\mathcal{C}_{n}^{-}) $ and recall \eqref{En} so that we can evolve it by $ \mathbb{H}_{n+1}(0) $, only to find
\begin{equation}\label{Edote}
E_{n} \dot{e}_{nM+q} = \begin{pmatrix}
- R_{n}^{+} L_{n} & R_{n}^{+} \\
- R_{n}^{T} - \bar{R}_{n+1} R_{n}^{+} L_{n} & \bar{R}_{n+1} R_{n}^{+}
\end{pmatrix} \begin{pmatrix}
-U_{2} \iota_{nM} \\
-L_{n} U_{2} \iota_{nM} 
\end{pmatrix} = 0
\end{equation}
because $ R_{n} = U_{1} \Sigma_{r} V_{1}^{T} $ and so $ R_{n}^{T} U_{2} = V_{1} \Sigma_{r} U_{1}^{T} U_{2} $, which is annihilated by $ U_{1}^{T} U_{2} = 0 $. Now, \eqref{Edote} tells us that $ \mathbb{H}_{n+1}(0) z_{n} = 0 $. As a result, the image of $ \mathcal{N}_{\omega}(\mathcal{C}_{n}^{-}) $ under the Hamiltonian evolution map $ \mathbb{H}_{n+1}(\lambda_{n+1}) $ is a single point $ \mathbb{H}_{n+1}(\lambda_{n+1}) \mathcal{N}_{\omega}(\mathcal{C}_{n}^{-}) = \{ F_{n+1} \lambda_{n+1} \} $ of the affine space $ \mathcal{C}_{n+1}^{+} $. In particular, we have $ \mathbb{H}_{n+1}(0) \mathcal{N}_{\omega}(\mathcal{C}_{n}^{-}) = \{ 0 \} $. From the conservation of the symplectic structure, it follows that $ \mathbb{H}_{n+1}(0) \mathcal{N}_{\omega}(\mathcal{C}_{n}^{-}) = \mathcal{N}_{\omega}(\mathbb{H}_{n+1}(0) \mathcal{C}_{n}^{-}) = \mathcal{N}_{\omega}(\mathcal{C}_{n+1}^{+(0)}) $. Therefore, we get that $ \mathcal{N}_{\omega}(\mathcal{C}_{n+1}^{+(0)}) = \{ 0 \} $, i.e., the subspace $ \mathcal{C}_{n+1}^{+(0)} $ of the post-constraint surface $ \mathcal{C}_{n+1}^{+} $ turns out to be symplectic.\\

Next we would like to evolve the vectors in $ \dot{\mathcal{C}}_{n}^{-} $. For this purpose we define a set $ \{ \ddot{e}_{n+1X} \} $, again with $ X = 1,..., r_{n}, q+1, ..., q+r_{n} $, by
\begin{equation}\label{ddotex}
\ddot{e}_{n+1X} = \mathbb{H}_{n+1}(0) \dot{e}_{nX} = E_{n} \dot{e}_{nX}
\end{equation}
This results in
\begin{equation}\label{ddoteE}
\ddot{e}_{n+1E} = \begin{pmatrix}
R_{n}^{+} U_{1} \varepsilon_{nE} \\
\bar{R}_{n+1} R_{n}^{+} U_{1} \varepsilon_{nE}
\end{pmatrix} = \begin{pmatrix}
V_{1} \Sigma_{r}^{-1} \varepsilon_{nE} \\
\bar{R}_{n+1} V_{1} \Sigma_{r}^{-1} \varepsilon_{nE}
\end{pmatrix}
\end{equation}
\begin{equation}\label{ddoteEq}
\ddot{e}_{n+1E+q} = \begin{pmatrix}
0 \\
R_{n}^{T} U_{1} \varepsilon_{nE}
\end{pmatrix} = \begin{pmatrix}
0 \\
V_{1} \Sigma_{r} \varepsilon_{nE}
\end{pmatrix}
\end{equation}
with $ E = 1, ..., r_{n} $. Here we again recalled $ R_{n} = U_{1} \Sigma_{r} V_{1}^{T} $ and $ R_{n}^{+} = V_{1} \Sigma_{r}^{-1} U_{1}^{T} $, which gives $ R_{n}^{T} U_{1} = V_{1} \Sigma_{r} $ and $ R_{n}^{+} U_{1} = V_{1} \Sigma_{r}^{-1} $. One can see that the set $ \{ \ddot{e}_{n+1X} \} $ is linearly independent. From the conservation of the symplectic structure
\begin{equation}\label{key}
\sigma_{XY} = \omega_{n}( \dot{e}_{nX}, \dot{e}_{nY} ) = \omega_{n+1}( \ddot{e}_{n+1X}, \ddot{e}_{n+1Y} )
\end{equation}
it subsequently follows that $ \{ \ddot{e}_{n+1X} \} $ is a symplectic basis of $ \mathbb{H}_{n+1}(0) \dot{\mathcal{C}}_{n}^{-} $. Since the null space $ \mathcal{N}_{\omega}(\mathcal{C}_{n}^{-}) $ was mapped to zero by $ \mathbb{H}_{n+1}(0) $ and there is nothing else to map, we get $ \mathbb{H}_{n+1}(0) \dot{\mathcal{C}}_{n}^{-} = \mathbb{H}_{n+1}(0) \mathcal{C}_{n}^{-} = \mathcal{C}_{n+1}^{+(0)} $. At the beginning of this paragraph, we have decided that $ \dot{\mathcal{C}}_{n+1}^{+} = \mathbb{H}_{n+1}(0) \dot{\mathcal{C}}_{n}^{-} $. It follows $ \dot{\mathcal{C}}_{n+1}^{+} = \mathcal{C}_{n+1}^{+(0)} $. This makes perfect sense: $ \mathcal{C}_{n+1}^{+(0)} $ is a linear, symplectic subspace of $ \mathcal{C}_{n+1}^{+} $ such that $ \mathcal{C}_{n+1}^{+(0)} \oplus \Lambda_{n+1} = \mathcal{C}_{n+1}^{+} $ and we will see in a moment that $ \Lambda_{n+1} = \mathcal{N}_{\omega}(\mathcal{C}_{n+1}^{+}) $ so that the direct sum complies to Observation \ref{cdotoplusn}. In summary, we have specified the two advertised representative spaces $ \dot{\mathcal{C}}_{n}^{-}, \dot{\mathcal{C}}_{n+1}^{+} $ as well as the \textit{symplectic} Hamiltonian map $ \mathbb{H}_{n+1}(0): \dot{\mathcal{C}}_{n}^{-} \rightarrow \dot{\mathcal{C}}_{n+1}^{+} $.\\

Our last task is to extend $ \{ \ddot{e}_{n+1X} \} $ to a symplectic basis $ \{ \ddot{e}_{n+1I} \}_{I=1}^{2q} $ of $ \mathcal{P}_{n+1} $, which we know is possible. The extension should span the point-space $ \Lambda_{n+1} $ as well as the rest of $ \mathcal{P}_{n+1} $, but there is not much information on how it should look. Therefore we have no other choice but to employ our creativity. Our strategy for finding the extension is the following: we again assume the transformation
\begin{equation}\label{ddoteWe}
e_{n+1J} = \ddot{e}_{n+1I} \ddot{W}_{n+1IJ}
\end{equation}
and write down everything we know about the matrix $ \ddot{W}_{n+1}^{-1} $ into the block form
\begin{equation}\label{ddotWgen}
\ddot{W}_{n+1}^{-1} = \begin{pmatrix}
V_{1} \Sigma_{r}^{-1} &  A  & 0 & B \\
\bar{R}_{n+1} V_{1} \Sigma_{r}^{-1} & C & V_{1}\Sigma_{r} & D
\end{pmatrix}
\end{equation}
with sought-for $ q \times s_{n} $ matrices $ A, B, C, D $. We require that $ \ddot{W}_{n+1} $ is symplectic, which implies a set of conditions on the unknown matrices through Theorem \ref{sympl}. One finds a class of solutions to these conditions parametrized as
\begin{equation}\label{ABCD}
\begin{aligned}
A &= V_{2} b \\
B &= V_{2} a \\
C &= \bar{R}_{n+1} V_{2} b  - V_{2} d\\
D &= \bar{R}_{n+1} V_{2} a  - V_{2} c
\end{aligned}
\end{equation}
where $ a, b, c, d $ are $ s_{n} \times s_{n} $ matrices such that
\begin{equation}\label{key}
w \equiv \begin{pmatrix}
a & b \\
c & d
\end{pmatrix}
\end{equation}
is symplectic.\\

Now, there are multiple choices of $ w $ which could be considered convenient, but we found after some trial and error that our purpose is best served by
\begin{equation}\label{ourw}
w = \begin{pmatrix}
\phantom{-}0 & \mathbf{1} \\
-\mathbf{1} & 0
\end{pmatrix}
\end{equation}
which yields perhaps the most simple and natural form of the matrix \eqref{ddotWgen}. The rest of the basis is then fixed to
\begin{equation}\label{ddoteM}
\ddot{e}_{n+1M} = \begin{pmatrix}
V_{2} \iota_{nM} \\
\bar{R}_{n+1} V_{2} \iota_{nM}
\end{pmatrix}, \qquad \ddot{e}_{n+1M+q} = \begin{pmatrix}
0 \\
V_{2} \iota_{nM}
\end{pmatrix}
\end{equation}
which reminds of the structure encountered in \eqref{doteM}. Now, recalling \eqref{Fn+1}, one may observe that the set $ \{ \ddot{e}_{n+1M} \}_{M = r_{n+1}}^{q} $ is in fact a basis of the point-space $ \Lambda_{n+1} $, whose elements are
\begin{equation}\label{Flamddot}
F_{n+1} \lambda_{n+1} = \begin{pmatrix}
V_{2} \\
\bar{R}_{n+1} V_{2}
\end{pmatrix} \lambda_{n+1} = \ddot{e}_{n+1M} \lambda_{n+1 M-r_{n}}
\end{equation}
At this point we can prove the preconceived relation $ \mathcal{N}_{\omega}(\mathcal{C}_{n+1}^{+}) = \Lambda_{n+1} $ with ease, by noticing that the general properties of symplectic bases imply $ \omega_{n+1}(\ddot{e}_{n+1 N}, \ddot{e}_{n+1 M}) = 0 $ and $ \omega_{n+1}(\ddot{e}_{n+1 N}, \ddot{e}_{n+1 X}) = 0 $ for all $ N, M = r_{n+1}, ..., q $ and all $  X = 1,..., r_{n}, q+1, ..., q+r_{n} $. These prove the point, since $ \ddot{e}_{n+1 X} $ form the basis of $ \mathcal{C}_{n+1}^{+(0)} $ and $ \ddot{e}_{n+1 M} $ form the basis of $ \Lambda_{n+1} $, as we know.  The reader can check the above symplectic products explicitly by plugging from \eqref{ddoteE}, \eqref{ddoteEq} and \eqref{ddoteM} and using the fact that $ V_{1}^{T}V_{2} = 0 $.\\

We can rewrite \eqref{ddotWgen} with \eqref{ABCD} and \eqref{ourw} as
\begin{equation}\label{ddotW}
\ddot{W}_{n+1}^{-1} = \begin{pmatrix}
V_{1} \Sigma_{r}^{-1} & V_{2} & 0 & 0  \\
\bar{R}_{n+1} V_{1} \Sigma_{r}^{-1} & \bar{R}_{n+1} V_{2} & V_{1}\Sigma_{r} &  V_{2}
\end{pmatrix}
\end{equation}
where it must be still understood that $ V_{1} \equiv V_{1}(R_{n}) $, $ V_{2} \equiv V_{2}(R_{n}) $ and $ \Sigma_{r} \equiv \Sigma_{r}(R_{n}) $. We repeat for clarity that the odd block columns of \eqref{ddotW} have width $ r_{n} $ and the even block columns have width $ s_{n} $. We can further simplify the form of \eqref{ddotW} by introducing the $ q \times q $ matrix
\begin{equation}\label{key}
\bar{\Sigma} \equiv \begin{pmatrix}
\Sigma_{r} & 0 \\
0 & \mathbf{1}
\end{pmatrix}
\end{equation}
With that, we can write
\begin{equation}\label{ddotW3}
\ddot{W}_{n+1}^{-1} = \begin{pmatrix}
V \bar{\Sigma}^{-1} & 0 \\
\bar{R}_{n+1} V \bar{\Sigma}^{-1} & V \bar{\Sigma}
\end{pmatrix}
\end{equation}
The inverse is easily found to be
\begin{equation}\label{ddotWorig}
\ddot{W}_{n+1} = \begin{pmatrix}
\bar{\Sigma} V^{T} & 0 \\
-\bar{\Sigma}^{-1} V^{T} \bar{R}_{n+1} & \bar{\Sigma}^{-1} V^{T}
\end{pmatrix}
\end{equation}
Eventually, we have at our hand the whole new symplectic basis $ \{ \ddot{e}_{n+1I} \}_{I=1}^{2q} $ of $ \mathcal{P}_{n+1} $ as desired.\\

\vspace{\baselineskip}

Let us summarize our findings and explain their significance for the coordinate description of the one-step evolution. In doing so, we will continue to view the post-constraint surface $ \mathcal{C}_{n+1}^{+} $ as a linear subspace of $ \mathcal{P}_{n+1} $. As we have made clear before, this is possible thanks to the fact that $ \Lambda_{n+1} $ has zero intersection with $ \mathcal{C}_{n+1}^{+(0)} $. We remind that we have fixed the representative space $ \dot{\mathcal{C}}_{n}^{-} $ to be the space spanned by $ \{ \dot{e}_{n E}, \dot{e}_{n E+q} \}_{E = 1}^{r_{n}} $. On the other hand, the representative space on the post-constraint surface was fixed as $ \dot{\mathcal{C}}_{n+1}^{+} = \mathcal{C}_{n+1}^{+(0)} $, which is spanned by $ \{ \ddot{e}_{n+1 E}, \ddot{e}_{n+1 E+q} \}_{E = 1}^{r_{n}} $. We also know that $ \mathcal{N}_{\omega}(\mathcal{C}_{n+1}^{+}) = \Lambda_{n+1} $ and that $ \mathcal{C}_{n+1}^{+(0)} \oplus \Lambda_{n+1} = \mathcal{C}_{n+1}^{+} $. As a linear subspace of $ \mathcal{P}_{n+1} $, the post-constraint surface $ \mathcal{C}_{n+1}^{+} $ has basis $ \{ \ddot{e}_{n+1 E}, \ddot{e}_{n+1 E+q} \}_{E = 1}^{r_{n}} \cup \{ \ddot{e}_{n+1 M} \}_{M = r_{n}+1}^{q} $.\\

We can now pass to the adapted coordinates and use them to describe vectors in $ \mathcal{P}_{n} $ and $ \mathcal{P}_{n+1} $. A vector $ u_{n} \in \mathcal{P}_{n} $ can be written in coordinates w.r.t. the new basis $ \{ \dot{e}_{nI} \}_{I = 1}^{2q} $ as $ u_{n} = \dot{u}_{nI} \dot{e}_{nI} $. Similarly, a vector $ v_{n+1} \in \mathcal{P}_{n+1} $ can be written in coordinates w.r.t. the new basis $ \{ \ddot{e}_{n+1I} \}_{I = 1}^{2q} $ as $ v_{n+1} = \ddot{v}_{n+1I} \ddot{e}_{n+1I} $.  Thanks to our construction, it is now exceptionally easy to make judgments about their nature:

\begin{observation}\label{ob:ad1}
	For $ u_{n} \in \mathcal{P}_{n} $, the following statements hold.
	\begin{enumerate}[{(i)}]
		\item $ u_{n} \in \mathcal{C}_{n}^{-} $ if and only if $ \dot{u}_{nM} = 0 $ for all $ M = r_{n}+1, ..., q $
		\item $ u_{n} \in \mathcal{N}_{\omega}(\mathcal{C}_{n}^{-}) $ if and only if $ \dot{u}_{nE} = \dot{u}_{nE+q} = \dot{u}_{nM} = 0 $ for all $ E = 1,..., r_{n} $ and all $ M = r_{n}+1, ..., q $ ~(i.e., only $ \dot{u}_{nM+q} $ can be nonzero)
		\item $ u_{n} \in \dot{\mathcal{C}}_{n}^{-} $ if and only if $ \dot{u}_{nM} = \dot{u}_{nM+q} = 0 $ for all $ M = r_{n}+1, ..., q $
	\end{enumerate}
\end{observation}

\begin{observation}\label{ob:ad2}
	For $ v_{n+1} \in \mathcal{P}_{n+1} $, the following statements hold.
	\begin{enumerate}[{(i)}]
		\item $ v_{n+1} \in \mathcal{C}_{n+1}^{+} $ if and only if $ \ddot{v}_{n+1M} = 0 $ for all $ M = r_{n}+1, ..., q $
		\item $ \mathcal{N}_{\omega}(\mathcal{C}_{n+1}^{+}) = \Lambda_{n+1} $, i.e., $ v_{n+1} \in \mathcal{N}_{\omega}(\mathcal{C}_{n+1}^{+}) $ if and only if $ \ddot{v}_{n+1E} = \ddot{v}_{n+1E+q} = \ddot{v}_{n+1M+q} = 0 $ for all $ E = 1,..., r_{n} $ and all $ M = r_{n}+1, ..., q $
		\item $ v_{n+1} \in \dot{\mathcal{C}}_{n+1}^{+} $ if and only if $ \ddot{v}_{n+1M} = \ddot{v}_{n+1M+q} = 0 $ for all $ M = r_{n}+1, ..., q $
	\end{enumerate}
\end{observation}
\begin{proof}
	Follows directly from the preceding discussion.
\end{proof}

\vspace{\baselineskip}

As for the general form of the evolution map, we know that all vectors $ u_{n} \in \mathcal{C}_{n}^{-} $ can be evolved into $ u_{n+1} = \mathbb{H}_{n+1}(\lambda_{n+1}) u_{n} \in \mathcal{C}_{n+1}^{+} $. In canonical coordinates, this is represented (in matrix form) as
\begin{equation}\label{un+1}
u_{n+1} =  E_{n} u_{n} + F_{n+1} \lambda_{n+1}
\end{equation}
On the other hand, upon using our adapted coordinates---recall the relations \eqref{ddotex} and \eqref{Flamddot}---, this prescription simplifies substantially. In particular, the vector components associated to the symplectic bases $ \{ \dot{e}_{nX} \} $ and $ \{ \ddot{e}_{n+1X} \} $, $ X = 1,..., r_{n}, q+1, ..., q+r_{n} $, of $ \dot{\mathcal{C}}_{n}^{-} $ and $ \dot{\mathcal{C}}_{n+1}^{+} $, respectively, are conserved by the evolution,
\begin{equation}\label{uddotX}
\ddot{u}_{n+1X} = \dot{u}_{nX}
\end{equation}
while the zero components $ \dot{u}_{nM} $ are updated with an arbitrary constant contribution from the point-set part $ F_{n+1} \lambda_{n+1} \in \Lambda_{n+1} $,
\begin{equation}\label{uddotM}
\ddot{u}_{n+1M} = \lambda_{n+1 M-r_{n}}
\end{equation}
and the null-space components $ \dot{u}_{nM+q} $ are annihilated,
\begin{equation}\label{uddotMq}
\ddot{u}_{n+1M+q} = \dot{u}_{nM} = 0
\end{equation}
Let us remark that we could have introduced a more logical transformation which would instead result in $ \ddot{u}_{n+1M} = \dot{u}_{nM} = 0 $ and $ \ddot{u}_{n+1M+q} = \lambda_{n+1 M-r_{n}} $. However, this transformation---try to use an identity matrix instead of \eqref{ourw} to see it come out---would have a less practical matrix \eqref{ddotWorig}, so we decided to proceed this way instead. Because of providing the Hamiltonian evolution with such a beautifully simple form, the adapted coordinates defined solemnly by the two symplectic matrices \eqref{dotW} and \eqref{ddotWorig} can be very helpful not only in classifying vectors, but also in describing the evolution. With this we close our discussion of the constraint surfaces.\\

\vspace{\baselineskip}

\subsection{Global Solutions}\label{globsol}
In this section we briefly discuss solutions spanning over the whole considered time interval from $ n = 0 $ to $ n = t $, i.e., elements of the solution space $ \mathcal{S}_{0t} $. Recall that a \textit{solution} $ y \in \mathcal{S}_{0t} $ is a point in $ \mathcal{P}_{0t} $ which satisfies momentum-matching $ ^{-}p_{n} = ~ ^{+}p_{n} $ as well as all the constraints originating in the irregularity of the system. At every time-slice $ n $ we can identify two kinds of pre-constraints: there is the forward pre-constraint $ C_{n} y_{n} = 0 $ which has to be satisfied by $ y_{n} $ so that the solution continues to time-slice $ n+1 $, and there is the backwards pre-constraint $ \bar{C}_{n} y_{n} = 0 $ which has to be satisfied should the solution continue to time-slice $ n-1 $.\\

When we talk about global solutions, the notion of pre-constraint surface is not sufficient: even if the time-slice data $ y_{n} \in \mathcal{P}_{n} $ satisfy the pre-constraint, i.e., $ y_{n} \in \mathcal{C}_{n}^{-} $, there is no guarantee that the evolved configuration $ y_{n+1} $ will be in $ \mathcal{C}_{n+1}^{-} $. We therefore define the constraint surfaces
\begin{equation}\label{Dn}
\mathcal{D}_{n} = \{ y_{n} \in \mathcal{P}_{n} ~ \vert ~ \exists \text{ solution } y \in \mathcal{S}_{0t} \text{ such that } y_{n} = \mathbb{P}_{n} y \}
\end{equation}
For a linear system like ours, one can check that $ \mathcal{D}_{n} $ are linear subspaces of $ \mathcal{P}_{n} $. We may also consider the total constraint space $ \mathcal{D}_{0t} = \bigotimes_{n = 0}^{t} \mathcal{D}_{n} $. By definition, each solution is in $ \mathcal{D}_{0t} $ but not all points in $ \mathcal{D}_{0t} $ are solutions, i.e., $ \mathcal{S}_{0t} \subset \mathcal{D}_{0t} $. We must keep in mind that because of the free parameters of the Hamiltonian evolution map, $ y \in \mathcal{S}_{0t} $ is in general not uniquely defined by $ y_{n} \in \mathcal{D}_{n} $.\\

The previously given argument for conservation of symplectic product can be extended by induction to arbitrary combination of evolution steps. We can be therefore sure that if $ y, z $ are two solutions, then $ \omega_{n}(y_{n},z_{n}) $ is independent of $ n $. This motivates a definition of \textit{product of solutions} $ \omega: \mathcal{S}_{0t} \times \mathcal{S}_{0t} \rightarrow \mathbb{R} $ with $ \omega(y,z) = \omega_{n}( y_{n} , z_{n}) $ for an arbitrary $ n \in \{0, ..., t\} $. This product is \textit{not} generally symplectic.\\

Having established $ \omega $ on $ \mathcal{S}_{0t} $, we can treat $ \mathcal{S}_{0t} $ in the same way we treated an arbitrary subspace $ \mathcal{C} $ of a symplectic space $ (\mathcal{P}, \omega) $ in Sec. \ref{constrsurf} and classify the solutions by their product structure. Let us say that two solutions $ y, \tilde{y} \in \mathcal{S}_{0t} $ are \textit{symplectically equivalent} if $ \omega(y,z) = \omega(\tilde{y},z) $ for all $ z \in \mathcal{S}_{0t} $, and write $ y \sim \tilde{y} $. Then we render the equivalence classes $ [y] $ of all symplectically equivalent solutions $ [y] = \{ \tilde{y} ~ \vert ~ \tilde{y} \sim y \} $. The space of such equivalence classes shall be denoted by $ \tilde{\mathcal{S}}_{0t} $. There is a naturally induced product $ \omega : \tilde{\mathcal{S}}_{0t} \times \tilde{\mathcal{S}}_{0t} \rightarrow \mathbb{R} $, $ \omega( [y], [z] ) = \omega(y,z) $. This is worth the effort for the following reason:

\begin{observation}\label{O1}
	The space $ (\tilde{\mathcal{S}}_{0t}, \omega) $ is symplectic.
\end{observation}
\begin{proof}
	The proof is analogical to that of Observation \ref{Octilde}.
\end{proof}

The construction of Sec. \ref{constrsurf} can be straightforwardly applied to $ \mathcal{D}_{n} $ which is a subspace of the symplectic space $ \mathcal{P}_{n} $. Thus we get the space $ \tilde{\mathcal{D}}_{n} $ of equivalence classes $ [y_{n}] = \{ \tilde{y}_{n} \in \mathcal{D}_{n} ~ \vert ~ \tilde{y}_{n} \sim y_{n} \} $ with $ y_{n} \sim \tilde{y}_{n} $ defined by $ \omega_{n}(y_{n},z_{n}) = \omega_{n}(\tilde{y}_{n},z_{n}) $ for all $ z_{n} \in \mathcal{D}_{n} $. The space $ \tilde{\mathcal{D}}_{n} $ is equal to $ \mathcal{D}_{n} / \mathcal{N}_{n} $ with $ \mathcal{N}_{n} = \{ z_{n} \in \mathcal{D}_{n} ~ \vert ~ z_{n} \sim 0 \} $. We of course set $ \omega_{n} : \tilde{\mathcal{D}}_{n} \times \tilde{\mathcal{D}}_{n} \to \mathbb{R} $ to act as $ \omega_{n}([y_{n}], [u_{n}]) = \omega_{n}(y_{n},u_{n}) $, resulting in the symplectic vector space $ (\tilde{\mathcal{D}}_{n}, \omega_{n}) $.\\

The relationship between $ \tilde{\mathcal{D}}_{n} $ and $ \tilde{\mathcal{S}}_{0t} $ is particularly simple:
\begin{observation}\label{O2}
	For every initial condition $ [y_{n}] \in \tilde{\mathcal{D}}_{n} $, $ n \in \{ 0, ..., t \} $ exists a solution $ [y] \in \tilde{\mathcal{S}}_{0t} $ such that $ [\mathbb{P}_{n} y] = [y_{n}] $. This solution is unique.
\end{observation}
\begin{proof}
	By definition of $ \mathcal{D}_{n} $, there is a solution $ y \in \mathcal{S}_{0t} $ for each $ y_{n} \in \mathcal{D}_{n} $ such that $ \mathbb{P}_{n} y = y_{n} $. Next, $ \omega(y,u) = \omega_{n}(\mathbb{P}_{n}y,\mathbb{P}_{n}u) $, therefore $ y \sim u \Leftrightarrow \mathbb{P}_{n}y \sim \mathbb{P}_{n}u $. It follows that $ [y_{n}] = [\mathbb{P}_{n} y] $. Assume there are two solutions $ [y], [u] \in  \tilde{\mathcal{S}} $ such that $ y_{n} \sim u_{n} $, then $ y \sim u $ and $ [y] = [u] $.
\end{proof}

We conclude our discussion by:

\begin{observation}\label{O3}
	The spaces $ (\tilde{\mathcal{D}}_{n}, \omega_{n}) $ for each $ n \in \{ 0, ..., t \} $ and $ (\tilde{\mathcal{S}}_{0t}, \omega) $ are all mutually symplectomorphic.
\end{observation}
\begin{proof}
	The symplectomorphism of $ (\tilde{\mathcal{D}}_{n}, \omega_{n}) $ (for arbitrary $ n $) and $ (\tilde{\mathcal{S}}_{0t}, \omega) $ is given by Observation \ref{O2}. Since the symplectomorphic relation is transitive, it follows that for any $ n,m \in \{ 0, ..., t \} $, $ (\tilde{\mathcal{D}}_{n}, \omega_{n}) $ is symplectomorphic to $ (\tilde{\mathcal{D}}_{m}, \omega_{m}) $.
\end{proof}

For practical purposes, we can go one more step and represent each class of symplectically equivalent solutions $ [y] \in \tilde{\mathcal{S}}_{0t} $ by a single solution $ y \in [y] $. The space of these representative solutions shall be denoted by $ \dot{\mathcal{S}}_{0t} $.  We require that $ \dot{\mathcal{S}}_{0t} $ is a linear subspace of $ \mathcal{S}_{0t} $. Once it is chosen, it fixes uniquely the spaces $ \dot{\mathcal{D}}_{n} = \{ \mathbb{P}_{n}y ~ \vert ~ y \in \dot{\mathcal{S}}_{0t} \} $ of the corresponding initial data. We let $ \dot{\mathcal{S}}_{0t} $,  $ \dot{\mathcal{D}}_{n} $ inherit the symplectic structures $ \omega $, $ \omega_{n} $ of $ \tilde{\mathcal{S}}_{0t} $, $ \tilde{\mathcal{D}}_{n} $, respectively. Note that $ (\tilde{\mathcal{S}}_{0t}, \omega) $ and $ (\dot{\mathcal{S}}_{0t}, \omega) $ are trivially symplectomorphic. Then $ (\dot{\mathcal{S}}_{0t}, \omega) $, $ (\dot{\mathcal{D}}_{n}, \omega_{n} ) $ become symplectic spaces and all the tildes in the statement of Observation \ref{O3} can be replaced by dots.\\

\vspace{\baselineskip}

\section{Massless Scalar Field on a 2D Spacetime Lattice}\label{toy}
In this section we look at a particular example of a discrete linear dynamical system and use it to demonstrate the application of the above introduced formalism. We consider massless scalar field on a Regge triangulation corresponding to a flat spacetime region with 1 space and 1 time dimension. This will truly be a toy model, since we keep the background fixed and only care about the field's dynamics. On the other hand, one should add that in two spacetime dimensions the Einstein tensor vanishes identically \cite{Strobl2000}, and this behavior carries over from continuum to lattice \cite{Hamber2009}, so it would suffice to consider a conformally flat spacetime to get the full theory including gravitation. However, this is not our objective now, as the present toy model will serve its illustrative purpose well.\\

According to our previous assumption, the triangulation shall be composed of a finite number of spacelike slices indexed by $ n \in \{0, ..., t\} $, such that every slice includes a finite number of vertices (at most $ q $) and every vertex is a member of exactly one slice. For simplicity, we shall consider triangulation with only two kinds of edges: spacelike and timelike, and suppose that all edges of each of these families have identical geometry. Edges between vertices which belong to the same slice are \textit{spacelike}, while edges between vertices which inhabit neighboring slices are \textit{timelike}. We do not allow for any other kind of edge.\\

Let us say more about the scalar field. One can describe it easily by associating a field value $ \varphi_{i} \in \mathbb{R} $ to every vertex $ i $. The corresponding scalar field action can be found e.g. in Sec. 6.12 of \cite{Hamber2009}; in our case it will take the form
\begin{equation}\label{sfaction}
S_{0t} = \frac{1}{2} \sum_{\text{edges } ij} w_{ij} ~ (\varphi_{i}-\varphi_{j})^{2}
\end{equation}
Here the sum  runs over all edges $ ij $ in the triangulation. We do not endow edges with any orientation and assume that $ ij $ is the same edge as $ ji $, so we can equivalently write $ \sum_{\text{edges } ij} \equiv \frac{1}{2} \sum_{ij} \delta^{e}_{ij} $ with
\begin{equation}\label{key}
\delta^{e}_{ij} = \begin{cases}
1 & \text{if the vertices } i \text{ and } j \text{ are connected by an edge} \\
0 & \text{otherwise}
\end{cases}
\end{equation}
We assume that edges only connect distinct vertices, i.e., $ \delta^{e}_{ii} = 0 $ for all $ i $. We remark that the action \eqref{sfaction} is similar to that used in Example 2.1. of \cite{Dittrich2013}, with the major difference that here we consider Lorentzian lattice, and not Euclidean. This is the reason why we need to include a coefficient $ w_{ij} $ providing \textit{weight} to every edge. It is proportional to the dual edge volume (here, area) and inversely proportional to the squared edge length:
\begin{equation}\label{key}
w_{ij} = \frac{\mathcal{V}_{ij}}{\mathcal{l}_{ij}^{2}}
\end{equation}
By definition, $ w_{ij} = w_{ji} $ is symmetric.\\

We shall assume for simplicity that all our triangles are identical. In result, there will be only two kinds of triangles in our lattice: (2,1) type triangles, which have two vertices at the sooner time-slice and one vertex at the later time-slice, and (1,2) type triangles whose configuration is the opposite. Note that all triangles have one spacelike edge and two timelike edges, regardless of their type. The dual edge volume is also the same for both the types. It may be fixed as $ \mathcal{V}_{ij} = m A $ where $ A $ is a constant contribution from one triangle (i.e., $ 1/3 $ of its area) and $ m $ is the number of triangles which contain the given edge. Our triangulation will be mostly built of \textit{interior} edges which belong to exactly two triangles; we therefore decide to divide the action by the overall constant $ 2A $. Occasionally it will be useful to consider \textit{boundary} edges which belong to only one triangle (typically edges on a boundary); for these we shall include a factor of $ 1/2 $ in $ w_{ij} $ to have the numbers right.\\

Because of the Lorentzian nature of our lattice, the squared edge lengths $ \mathcal{l}_{ij}^{2} $ must be taken in account too. In a triangulation of a flat spacetime, they are given simply by (squared) spacetime intervals between points which correspond to the two vertices of the edge in question. We shall fix them as follows. First we provide our flat spacetime region with an orthogonal frame consisting of a time coordinate $ t $ and a space coordinate $ x $. Then we draw a triangular lattice such that every triangle has one edge aligned with the $ x $ direction and all triangles are equilateral with unit side in the Euclidean metric induced by the frame, see Fig. \ref{fig:triang}.

\begin{figure}[h!]
	\centering
	\begin{tikzpicture}[scale=1]
	\tikzset{
		edget/.style={
			draw,dashed,color=lightgray!100,line width=0.3mm
		},
		edges/.style={
			draw,-,color=lightgray!100,line width=0.3mm
		}
	}
	
	\begin{axis}[
	axis lines=left,
	xtick=\empty,
	ytick=\empty,
	axis line style={->},
	tick style={color=black},
	xlabel={$ x $},
	ylabel={$ t $},
	xmin=0,
	xmax=4,
	ymin=0,
	ymax=3.2,
	y=10mm,
	x=10mm,
	]
	
	\coordinate (c0) at (axis cs: 0,0);
	\coordinate (c1) at (axis cs: 1,0);
	\coordinate (c2) at (axis cs: 2,0);
	\coordinate (c3) at (axis cs: 3,0);
	
	\coordinate (cc) at (axis cs: 0,0.866);
	\coordinate (c4) at (axis cs: 0.5,0.866);
	\coordinate (c5) at (axis cs: 1.5,0.866);
	\coordinate (c6) at (axis cs: 2.5,0.866);
	\coordinate (c7) at (axis cs: 3.5,0.866);
	
	\coordinate (c8) at (axis cs: 0,1.732);
	\coordinate (c9) at (axis cs: 1,1.732);
	\coordinate (c10) at (axis cs: 2,1.732);
	\coordinate (c11) at (axis cs: 3,1.732);	
	
	\coordinate (cg) at (axis cs: 0,2.598);
	\coordinate (c12) at (axis cs: 0.5,2.598);
	\coordinate (c13) at (axis cs: 1.5,2.598);
	\coordinate (c14) at (axis cs: 2.5,2.598);
	\coordinate (c15) at (axis cs: 3.5,2.598);

	\draw[edget] (c0) -- (c4) -- (c1) -- (c5) -- (c2) -- (c6) -- (c3) -- (c7);
	\draw[edges] (cc) -- (c4) -- (c5) -- (c6) -- (c7);
	\draw[edget] (c8) -- (c4) -- (c9) -- (c5) -- (c10) -- (c6) -- (c11) -- (c7);
	\draw[edges] (c8) -- (c9) -- (c10) -- (c11);
	\draw[edget] (c8) -- (c12) -- (c9) -- (c13) -- (c10) -- (c14) -- (c11) -- (c15);
	\draw[edges] (cg) -- (c12) -- (c13) -- (c14) -- (c15);
	
	\end{axis}
	
	\end{tikzpicture}
	\vspace{0 mm}
	\caption{Triangulation of flat 2-dimensional spacetime. Spacelike edges are drawn in full line, timelike edges are drawn in dashed line.}
	\label{fig:triang}
	\vspace{4 mm}
\end{figure}
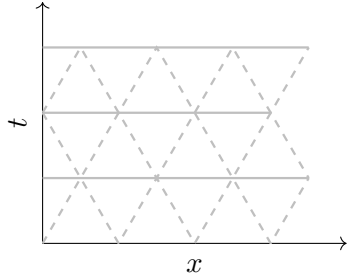

It follows that all spacelike edges have squared length $ 1 $ and all timelike edges have squared length $ -1/2 $, i.e., our distinction of the edges has the expected geometrical meaning. Altogether, we put
\begin{equation}\label{wijint}
w_{ij} = \begin{cases}
1 & \text{if } ij \text{ is an interior spacelike edge} \\
-2 & \text{if } ij \text{ is an interior timelike edge} \\
1/2 & \text{if } ij \text{ is a boundary spacelike edge} \\
-1 & \text{if } ij \text{ is a boundary timelike edge}
\end{cases}
\end{equation}
If there is no edge between the vertices $ i $ and $ j $, the weight $ w_{ij} $ is undefined.\\

The last thing we need to decide before we go on to employ the discrete canonical evolution is the topology of the lattice. Since we need each time-slice to have a limited number of vertices, we opt for the tube topology with time direction along the tube. More precisely, we will consider a lattice in which each time-slice is formed by a single closed loop of vertices connected by spacelike edges, and the individual neighboring time-slices are connected by timelike edges so that the resulting lattice is triangular. Since we want to consider only a finite number of time-slices, we cut the tube so that it starts with time-slice $ 0 $ and ends with time-slice $ t $. Consequently, all the spacelike edges at time-steps $ 0 $ and $ t $ will be boundary edges. Then we can say that \eqref{sfaction} is the action of the field corresponding to the region between $ 0 $ and $ t $, as one would expect.\\

Suppose we are given an instance of the above defined triangular lattice with the scalar field living on it. In the general case, the number of vertices in a time-slice will be varying with time. As discussed in \cite{Dittrich2012}, one can use the formalism of discrete canonical evolution to describe the field's dynamics. Suppose that the largest time-slice has $ q $ vertices; then one simply provides all other time-slices with additional \textit{virtual} vertices so that every time-slice has $ q $ vertices in the end. After that, the total number of vertices is $ N \equiv q(t+1) $. Virtual vertices are not connected by any edges, and therefore have no influence on the geometry. If we include into our consideration the field values at virtual vertices, we can say that our field has $ q $ degrees of freedom at each time-slice. The obvious implication of this trick is that the field values at virtual vertices have no significance for the action \eqref{sfaction}. In consequence, the system at hand will be irregular.\\

Now we can straightforwardly apply the formalism of discrete canonical evolution. Since the action \eqref{sfaction} is quadratic in field values, the system satisfies our additional assumption of linearity, so we can take advantage of our treatment given in the preceding section. We rewrite
\begin{equation}\label{sfactioncondensed}
S_{0t} = \textstyle \frac{1}{2} ~ \varphi^{T} K \varphi
\end{equation}
where $ \varphi \in \mathbb{R}^{N} $ and $ K $ is a real, symmetric $ N \times N $ matrix which we call the \textit{dynamical matrix}. By comparison of \eqref{sfaction} and \eqref{sfactioncondensed}, one finds that
\begin{equation}\label{Kij}
K_{ij} = \delta_{ij} \sum_{k} \delta^{e}_{ik} w_{ik} - \delta^{e}_{ij} w_{ij}
\end{equation}
From this expression it is clear that all the row and column sums of $ K $ are zero. Let us also point out that since virtual vertices are by definition not associated to any edges, $ K_{ij} = 0 $ whenever $ i $ or $ j $ is a virtual vertex.\\

Next we need to split up the action into individual time-step contributions. We shall do that simply by splitting the lattice into $ t $ individual time-steps (between 0 and 1, ..., between $ t-1 $ and $ t $). The splitting of lattice induces a corresponding splitting of the matrix $ K $. We let $ K_{(n)} $ be the $ q \times q $ submatrix of $ K $ corresponding to time-slice $ n $ and $ K_{(n,n+1)} $ the $ q \times q $ submatrix of $ K $ with rows corresponding to variables at time-slice $ n $ and columns corresponding to variables at time-slice $ n+1 $. Thanks to the symmetry of $ K $, the submatrix $ K_{(n)} $ is symmetric and $ K_{(n,n+1)}^{T} = K_{(n+1,n)} $. Moreover, our splitting of the lattice results in further decomposition $ K_{(n)} = K_{(n)}^{-} + K_{(n)}^{+} $ (for $ n = 1, ..., t-1 $) where $ K_{(n)}^{-} $ and $  K_{(n)}^{+} $ describe the boundary time-slice $ n $ of the two separated time steps: one between $ n-1 $ and $ n $ ($ - $), other between $ n $ and $ n+1 $ ($ + $). These matrices are given by the same formula \eqref{Kij} (with $ i,j $ both belonging to $ n $) to which one plugs the lattice of the appropriate individual time-step. Note that the only quantities which change in splitting the lattice are the dual volumes of spacelike edges. An interior spacelike edge $ ij $ has dual volume $ \mathcal{V}_{ij} = \mathcal{V}_{ij}^{-} + \mathcal{V}_{ij}^{+} $ whose one part $ \mathcal{V}_{ij}^{-} $ lies in the time-slice between $ n-1 $ and $ n $ and the other part $ \mathcal{V}_{ij}^{+} $ lies in the time-slice between $ n $ and $ n+1 $. The splitting of $ K_{(n)} $ therefore corresponds to splitting of these dual volumes according to the given geometry, so that one gets the correct time-step action contribution $ S_{n+1} $. In our simplified setting, the splitting is done very easily by turning all (originally interior) spacelike edges into boundary spacelike edges, i.e., dividing their edge weight by the factor of two. See Fig. \ref{fig:split} for an illustration.\\

\begin{figure}[h!]
	\centering
	\begin{tikzpicture}[scale=1]
	\tikzset{
		edget/.style={
			draw,dashed,color=lightgray!100,line width=0.3mm
		},
		edges/.style={
			draw,-,color=lightgray!100,line width=0.3mm
		},
		edgess/.style={
			draw,-,color=lightgray!100,line width=0.8mm
		},
		arrow/.style={
			draw,->,color=black,line width=0.3mm
		},
	}
	
	\coordinate (c1) at (1,0);
	\coordinate (c2) at (2,0);
	\coordinate (c3) at (3,0);
	\coordinate (c4) at (0.5,0.866);
	\coordinate (c5) at (1.5,0.866);
	\coordinate (c6) at (2.5,0.866);
	\coordinate (c7) at (3.5,0.866);
	\coordinate (c9) at (1,1.732);
	\coordinate (c10) at (2,1.732);
	\coordinate (c11) at (3,1.732);
	
	\coordinate (a) at (4.5,0.866);
	\coordinate (b) at (5.5,0.866);
	
	\coordinate (d1) at (7,-0.1);
	\coordinate (d2) at (8,-0.1);
	\coordinate (d3) at (9,-0.1);
	\coordinate (d4) at (6.5,0.766);
	\coordinate (d5) at (7.5,0.76);
	\coordinate (d6) at (8.5,0.766);
	\coordinate (d7) at (9.5,0.766);
	\coordinate (e4) at (6.5,0.966);
	\coordinate (e5) at (7.5,0.966);
	\coordinate (e6) at (8.5,0.966);
	\coordinate (e7) at (9.5,0.966);
	\coordinate (d9) at (7,1.832);
	\coordinate (d10) at (8,1.832);
	\coordinate (d11) at (9,1.832);
	
	\draw[edges] (c1) -- (c2) -- (c3);
	\draw[edget] (c4) -- (c1) -- (c5) -- (c2) -- (c6) -- (c3) -- (c7);
	\draw[edgess] (c4) -- (c5) -- (c6) -- (c7);
	\draw[edget] (c4) -- (c9) -- (c5) -- (c10) -- (c6) -- (c11) -- (c7);
	\draw[edges] (c9) -- (c10) -- (c11);
	
	\draw[arrow] (a) -- (b);
	
	\draw[edges] (d1) -- (d2) -- (d3);
	\draw[edget] (d4) -- (d1) -- (d5) -- (d2) -- (d6) -- (d3) -- (d7);
	\draw[edges] (d4) -- (d5) -- (d6) -- (d7);
	
	\draw[edges] (e4) -- (e5) -- (e6) -- (e7);
	\draw[edget] (e4) -- (d9) -- (e5) -- (d10) -- (e6) -- (d11) -- (e7);
	\draw[edges] (d9) -- (d10) -- (d11);
	
	\end{tikzpicture}
	\vspace{0 mm}
	\caption{Illustration of the splitting procedure. Spacelike interior edges (drawn in bold line) turn into spacelike boundary edges (drawn in normal line). Timelike edges (drawn in dashed line) remain unchanged.}
	\label{fig:split}
	\vspace{4 mm}
\end{figure}
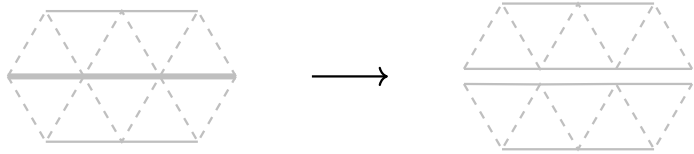

The action contribution from the time-step between $ n $ and $ n+1 $ takes the form
\begin{equation}\label{Sn+1}
\begin{aligned}
S_{n+1} = & \textstyle \frac{1}{2} \left( x_{n}^{T} K_{(n)}^{+} x_{n} + 2 x_{n}^{T} K_{(n,n+1)} x_{n+1} + x_{n+1}^{T} K_{(n+1)}^{-} x_{n+1} \right)
\end{aligned}
\end{equation} 
where $ x_{n} \in \mathbb{R}^{q} $ is the $ q $-tuple of field values $ \varphi_{i} $ in vertices $ i $ (including the virtual ones) belonging to time-slice $ n $. The matrix $ K_{(n,n+1)} $ describes the interaction along timelike edges so it does not come with any additional factor. With this, one easily identifies the matrices in \eqref{prepostmomentamatrix} as
\begin{equation}\label{LRK}
\begin{aligned}
L_{n}&= \textstyle K_{(n)}^{+} \qquad & R_{n}&=K_{(n,n+1)} \\
\bar{L}_{n+1}&= -R_{n}^{T} = -K_{(n+1,n)} \qquad & \bar{R}_{n+1}&= - K_{(n+1)}^{-} \\
\end{aligned}
\end{equation}
We see that \eqref{pardevsymm} and \eqref{pardevsymm2} indeed hold. The reader can also easily check that the individual contributions \eqref{Sn+1} give the action \eqref{sfactioncondensed},
\begin{equation}\label{actionadd}
S_{0t} = \sum_{n=0}^{t-1} S_{n+1}(x_{n},x_{n+1})
\end{equation}
as we desire. Let us remark that thanks to our assumption of a closed loop topology of each time-slice, there will be no timelike boundary, and all timelike edges will be found in the interior. On the other hand, due to our splitting of the lattice into individual time-steps, all spacelike edges will be on the boundary. Therefore \eqref{wijint} simplifies to
\begin{equation}\label{wijint2}
w_{ij} = \begin{cases}
-2 & \text{if } ij \text{ is an interior timelike edge} \\
1/2 & \text{if } ij \text{ is a boundary spacelike edge}
\end{cases}
\end{equation}\\

\subsection{One-Step Examples}
We can move on to discuss particular examples of lattice time-steps. We craft them so that they are as simple as possible and at the same time make visible the full range of the model's behavior. On the most basic level, there are three situations with different implications for the evolution. First, the vertices of the lattice are equally distributed amongst time-slices and well connected; then the system turns out regular. Second, the number of well connected vertices decreases in a time-step which results in a pre-constraint. Third, the number of well connected vertices increases in a time-step which results in a free parameter of the evolution. The first three examples are supposed to illustrate these cases. Last but not least, we should remark that the regularity of the evolution is not only dependent on the numbers of vertices at subsequent time-slices, but also on their connectivity. If the connectivity is poor, we intuitively feel that the system will be irregular, because the lattice will obstruct propagation of degrees of freedom. However, there are occasions on which our intuition can be misleading. We demonstrate this fact by one bonus example.\\

\begin{example}\label{ex:1}
	First we consider a time-step between time-slices 0 and 1 with exactly three vertices at each time-slice. The lattice is depicted in Fig. \ref{fig:lat1}.\\
	
	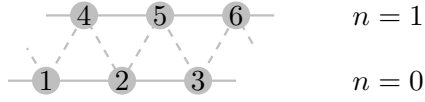
\begin{figure}[h!]
		\centering
		\begin{tikzpicture}[scale=1]
		\tikzset{
			vertex/.style={
				shape=circle,fill=lightgray!100,minimum size=3mm,inner sep=0.2mm, label={[fill=none,label distance=1mm]90:#1}
			},
			vertexwhite/.style={
				shape=circle,fill=white!100,minimum size=10mm,inner sep=0.2mm, label={[fill=none,label distance=1mm]90:#1}
			},
			edge/.style={
				draw,-,color=lightgray!100,line width=0.3mm
			},
			edget/.style={
				draw,dashed,color=lightgray!100,line width=0.3mm
			}
		}
		
		\coordinate (cia) at (-2,0);
		\coordinate (c1) at (-1,0);
		\coordinate (c2) at (0,0);
		\coordinate (c3) at (1,0);
		\coordinate (cib) at (2,0);
		\coordinate (cic) at (-1.5,0.866);
		\coordinate (c4) at (-0.5,0.866);
		\coordinate (c5) at (0.5,0.866);
		\coordinate (c6) at (1.5,0.866);
		\coordinate (cid) at (2.5,0.866);

		\draw[edge] (cia) -- (c1) -- (c2) -- (c3) -- (cib);
		\draw[edget] (cic) -- (c1) -- (c4) -- (c2) -- (c5) -- (c3) -- (c6) -- (cib);
		\draw[edge] (cic) -- (c4) -- (c5) -- (c6) -- (cid);
		
		\node[vertex] at (c1) {1};
		\node[vertex] at (c2) {2};
		\node[vertex] at (c3) {3};
		\node[vertex] at (c4) {4};
		\node[vertex] at (c5) {5};
		\node[vertex] at (c6) {6};
		\node[vertexwhite] at (cia) {};
		\node[vertexwhite] at (cib) {};
		\node[vertexwhite] at (cic) {};
		\node[vertexwhite] at (cid) {};
		\node[] at (3.5,0) {$ n = 0 $};
		\node[] at (3.5,0.866) {$ n = 1 $};

		
		\end{tikzpicture}
		\vspace{0 mm}
		\caption{Diagram of the time-step lattice of Example \ref{ex:1}. The fragments of edges on the right are meant to be connected to the fragments on the left, so that each time-slice is a closed loop.}
		\label{fig:lat1}
		\vspace{4 mm}
	\end{figure}

	We have $ q = 3 $, $ t = 1 $ and $ N = 6 $. The dynamical matrix \eqref{Kij} is
	\begin{equation}\label{key}
	K = \begin{pmatrix}
	-3 & -1/2 & -1/2 & 2 & 0 & 2 \\
	& -3 & -1/2 & 2 & 2 & 0 \\
	& & -3 & 0 & 2 & 2 \\
	& & & -3 & -1/2 & -1/2 \\
	& & & & -3 & -1/2 \\
	& & & & & -3 \\
	\end{pmatrix}
	\end{equation}
	Because the matrix is symmetric, we only write the upper triangle. One can check that the row and column sums of $ K $ are indeed zero. It is easy to read out the matrices of \eqref{LRK}. Since we have a single time-step, no splitting is needed. We get
	\begin{equation}\label{key}
	L_{0} = - \frac{1}{2} \begin{pmatrix}
	6 & 1 & 1 \\
	1 & 6 & 1 \\
	1 & 1 & 6
	\end{pmatrix}, \qquad R_{0} = \begin{pmatrix}
	2 & 0 & 2 \\
	2 & 2 & 0 \\
	0 & 2 & 2
	\end{pmatrix}, \qquad \bar{R}_{1} = - L_{0}
	\end{equation}
	At this point we can easily express the canonical evolution between time-slices 0 and 1 by plugging into \eqref{preconst}---\eqref{Fn+1}. Since $ R_{0} $ is regular, there is no pre-constraint, and the point-space $ \Lambda_{1} $ has dimension zero. In other words, the present single-time-step system is regular. For the evolution we get simply $ y_{1} = E_{0} y_{0} $ with
	\begin{equation}\label{key}
	E_{0} = \frac{1}{4} \begin{pmatrix}
	3 & 3 & -2 & 1 & 1 & -1 \\
	-2 & 3 & 3 & -1 & 1 & 1 \\
	3 & -2 & 3 & 1 & -1 & 1 \\
	3/2 & 3/2 & -3 & 3 & 3 & -2 \\
	-3 & 3/2 & 3/2 & -2 & 3 & 3 \\
	3/2 & -3 & 3/2 & 3 & -2 & 3 \\
	\end{pmatrix}
	\end{equation}
	This solves uniquely any initial-value problem. For instance, the canonical initial vector $ y_{0} = \begin{pmatrix} 1 & 0 & 0 & 0 & 0 & 0 \end{pmatrix}^{T} $ evolves into $ y_{1} = \frac{1}{4} \begin{pmatrix} 3 & -2 & 3 & 3/2 & -3 & 3/2 \end{pmatrix}^{T} $. The reader can check that the symplectic form is fully conserved.\\
	
	Eventually, let us switch to the adapted coordinates. First we perform the singular value decomposition $ R_{0} = U \Sigma V^{T} $ with the result
	\begin{equation}\label{key}
	U = \frac{1}{\sqrt{6}} \begin{pmatrix}
	\sqrt{2} & 0 & -2 \\
	\sqrt{2} & -\sqrt{3} & 1 \\
	\sqrt{2} & \sqrt{3} & 1
	\end{pmatrix}, \qquad \Sigma = \begin{pmatrix}
	4 & 0 & 0 \\
	0 & 2 & 0 \\
	0 & 0 & 2
	\end{pmatrix}, \qquad V = \frac{1}{\sqrt{6}} \begin{pmatrix}
	\sqrt{2} & -\sqrt{3} & -1 \\
	\sqrt{2} & 0 & 2 \\
	\sqrt{2} & \sqrt{3} & -1
	\end{pmatrix}
	\end{equation}
	From \eqref{dotW} and \eqref{ddotWorig} we have
	\begin{equation}\label{Wmatrices}
	\dot{W}_{0} = \begin{pmatrix}
	-U^{T} L_{0} & U^{T} \\
	-U^{T} & 0
	\end{pmatrix}, \qquad \ddot{W}_{1} = \begin{pmatrix}
	\bar{\Sigma} V^{T} & 0 \\
	-\bar{\Sigma}^{-1} V^{T} \bar{R}_{1} & \bar{\Sigma}^{-1} V^{T}
	\end{pmatrix}
	\end{equation}
	These give
	\begin{equation}\label{dotW0ex1}
	\dot{W}_{0} = \frac{1}{2\sqrt{6}} \begin{pmatrix}
	8 \sqrt{2} & 8 \sqrt{2} & 8 \sqrt{2} & 2 \sqrt{2} &  2 \sqrt{2} & 2 \sqrt{2}\\
	0 & - 5 \sqrt{3} & 5 \sqrt{3} & 0 &  -2 \sqrt{3} & 2 \sqrt{3}\\
	-10 &  5 & 5 & -4 &  2 & 2\\
	- 2 \sqrt{2} &  - 2 \sqrt{2} & - 2 \sqrt{2} & 0 & 0 & 0\\
	0 &  2 \sqrt{3} & - 2 \sqrt{3} & 0 & 0 & 0\\
	4 &  -2 & -2 & 0 & 0 & 0\\
	\end{pmatrix}
	\end{equation}
	and
	\begin{equation}\label{ddotW1ex1}
	\ddot{W}_{1} = \frac{1}{4 \sqrt{6}} \begin{pmatrix}
	16 \sqrt{2} & 16 \sqrt{2} & 16 \sqrt{2} & 0 & 0 & 0 \\
	-8 \sqrt{3} & 0 & 8 \sqrt{3} & 0 & 0 & 0 \\
	-8 & 16 & -8 & 0 & 0 & 0 \\
	- 4 \sqrt{2} & - 4 \sqrt{2} & - 4 \sqrt{2} & \sqrt{2} & \sqrt{2} & \sqrt{2}  \\
	5 \sqrt{3} & 0 & - 5 \sqrt{3} & - 2 \sqrt{3} & 0 & 2 \sqrt{3} \\
	5 & -10 & 5 & -2 & 4 & -2  \\
	\end{pmatrix}
	\end{equation}
	If we express the above vectors $ y_{0} $ and $ y_{1} $ in the adapted bases by transforming them with the matrices \eqref{dotW0ex1} and \eqref{ddotW1ex1}, we get $ \dot{y}_{0} = \dot{W}_{0} y_{0} = \frac{1}{\sqrt{6}} \begin{pmatrix} 4 \sqrt{2} & 0 & -5 & -\sqrt{2} & 0 & 2 \end{pmatrix}^{T} $ and $ \ddot{y}_{1} = \ddot{W}_{1} y_{1} = \dot{y}_{0} $, i.e., the adapted coordinates of the vector are conserved as expected.
\end{example}

\vspace{\baselineskip}

\begin{example}\label{ex:2}
	Now let us consider a different triangular lattice with three vertices at time-slice 0 but only one vertex at time-slice 1, as illustrated by Fig. \ref{fig:lat2}. We presume that because of the loss of degrees of freedom in the time-step from 0 to 1, the system will be irregular and a non-trivial pre-constraint will arise.\\
	
	\begin{figure}[h!]
		\centering
		\begin{tikzpicture}[scale=1]
		\tikzset{
			vertex/.style={
				shape=circle,fill=lightgray!100,minimum size=3mm,inner sep=0.2mm, label={[fill=none,label distance=1mm]90:#1}
			},
			vertexwhite/.style={
				shape=circle,fill=white!100,minimum size=10mm,inner sep=0.2mm, label={[fill=none,label distance=1mm]90:#1}
			},
			vertexvirtual/.style={
				shape=circle,draw=lightgray!100,fill=white!100,line width=0.3mm,minimum size=2.7mm,inner sep=0.2mm, 	label={[fill=none,label distance=1mm]90:#1}
			},
			edge/.style={
				draw,-,color=lightgray!100,line width=0.3mm
			},
			edget/.style={
				draw,dashed,color=lightgray!100,line width=0.3mm
			},
		}
		
		\coordinate (cia) at (-2,0);
		\coordinate (c1) at (-1,0);
		\coordinate (c2) at (0,0);
		\coordinate (c3) at (1,0);
		\coordinate (cib) at (2,0);
		\coordinate (cic) at (-2,0.866);
		\coordinate (c4) at (-1,0.866);
		\coordinate (c5) at (0,0.866);
		\coordinate (c6) at (1,0.866);
		\coordinate (cid) at (2,0.866);

		\draw[edge] (cia) -- (c1) -- (c2) -- (c3) -- (cib);
		\draw[edget] (c1) -- (c5) -- (c2);
		\draw[edget] (c3) -- (c5);
		
		\node[vertex] at (c1) {1};
		\node[vertex] at (c2) {2};
		\node[vertex] at (c3) {3};
		\node[vertexvirtual] at (c4) {4};
		\node[vertex] at (c5) {5};
		\node[vertexwhite] at (cia) {};
		\node[vertexwhite] at (cib) {};
		\node[vertexwhite] at (cic) {};
		\node[vertexwhite] at (cid) {};
		\node[vertexvirtual] at (c6) {6};
		\node[] at (3.5,0) {$ n = 0 $};
		\node[] at (3.5,0.866) {$ n = 1 $};
		
		\end{tikzpicture}
		\vspace{0 mm}
		\caption{Diagram of the time-step lattice of Example \ref{ex:2}. It is made of three identical type 2-1 triangles. Vertices 4 and 6 are virtual. Dashed edges are timelike.}
		\label{fig:lat2}
		\vspace{4 mm}
	\end{figure}
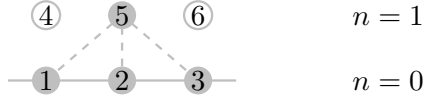
	
	We have $ q = 3 $, $ t = 1 $ and $ N = 6 $ as before. We read out the dynamical matrix
	\begin{equation}\label{key}
	K = \begin{pmatrix}
	-1 & -1/2 & -1/2 & 0 & 2 & 0 \\
	& -1 & -1/2 & 0 & 2 & 0 \\
	& & -1 & 0 & 2 & 0 \\
	& & & 0 & 0 & 0 \\
	& & & & -6 & 0 \\
	& & & & & 0 \\
	\end{pmatrix}
	\end{equation}
	and the matrices governing the evolution
	\begin{equation}\label{key}
	L_{0} = - \frac{1}{2} \begin{pmatrix}
	2 & 1 & 1 \\
	1 & 2 & 1 \\
	1 & 1 & 2
	\end{pmatrix}, \qquad R_{0} = \begin{pmatrix}
	0 & 2 & 0 \\
	0 & 2 & 0 \\
	0 & 2 & 0
	\end{pmatrix}, \qquad \bar{R}_{1} = \begin{pmatrix}
	0 & 0 & 0 \\
	0 & 6 & 0 \\
	0 & 0 & 0
	\end{pmatrix}
	\end{equation}
	As we presumed, $ R_{0} $ is singular. We see that $ r_{0} = 1 $ and $ s_{0} = 2 $. According to \eqref{preconst}, there is a pre-constraint $ C_{0} y_{0} = 0 $ which must be satisfied if we want to evolve $ y_{0} $ to the next time-slice. To express the pre-constraint, we may take advantage of the singular value decomposition $ R_{0} = U \Sigma V^{T} $ with
	\begin{equation}\label{key}
	U = \frac{1}{\sqrt{6}} \begin{pmatrix}
	\sqrt{2} & -\sqrt{3} & -1 \\
	\sqrt{2} & 0 & 2 \\
	\sqrt{2} & \sqrt{3} & -1
	\end{pmatrix}, \qquad \Sigma = \begin{pmatrix}
	2 \sqrt{3} & 0 & 0 \\
	0 & 0 & 0 \\
	0 & 0 & 0
	\end{pmatrix}, \qquad V = \begin{pmatrix}
	0 & 0 & 1 \\
	1 & 0 & 0 \\
	0 & 1 & 0
	\end{pmatrix}
	\end{equation}
	We use $ P_{\mathcal{N}(R_{n}^{T})} = U_{2} U_{2}^{T} $ where
	\begin{equation}\label{key}
	U_{2} = \frac{1}{\sqrt{6}} \begin{pmatrix}
	-\sqrt{3} & -1 \\
	0 & 2 \\
	\sqrt{3} & -1
	\end{pmatrix}, \qquad P_{\mathcal{N}(R_{n}^{T})} = - \frac{1}{3} \begin{pmatrix}
	-2 & 1 & 1 \\
	1 & -2 & 1 \\
	1 & 1 & -2
	\end{pmatrix}
	\end{equation}
	and compute
	\begin{equation}\label{key}
	C_{0} = - \frac{1}{6} \begin{pmatrix}
	-2 & 1 & 1 & -4 & 2 & 2 \\
	1 & -2 & 1 & 2 & -4 & 2 \\
	1 & 1 & -2 & 2 & 2 & -4
	\end{pmatrix}
	\end{equation}
	The pre-constraint surface $ \mathcal{C}_{0}^{-} $ is identified as the null space of this matrix; it has dimension four. This is because the total dimension of the phase space is six and the pre-constraint has dimension two. The latter corresponds to the number of virtual vertices. A state $ y_{0} $ can be evolved to time-slice 1 if and only if it belongs to $ \mathcal{C}_{0}^{-} $. The evolved state is never unique, since it is given by $ y_{1} = E_{0} y_{0} + F_{1} \lambda_{1} $ with an arbitrary $ \lambda_{1} \in \mathbb{R}^{2} $. A quick calculation reveals
	\begin{equation}\label{E0F0ex2}
	E_{0} = \frac{1}{6} \begin{pmatrix}
	0 & 0 & 0 & 0 & 0 & 0 \\
	2 & 2 & 2 & 1 & 1 & 1 \\
	0 & 0 & 0 & 0 & 0 & 0 \\
	0 & 0 & 0 & 0 & 0 & 0 \\
	0 & 0 & 0 & 6 & 6 & 6 \\
	0 & 0 & 0 & 0 & 0 & 0
	\end{pmatrix}, \qquad F_{1} = \begin{pmatrix}
	0 & 1 \\
	0 & 0 \\
	1 & 0 \\
	0 & 0 \\
	0 & 0 \\
	0 & 0
	\end{pmatrix}
	\end{equation}
	One can observe that $ F_{1} \lambda_{1} $ adds an arbitrary contribution to the field values at virtual vertices 4 and 6. This makes perfect sense, because the virtual vertices have no physical meaning, so it would be strange if their associated field values were in any way determined. On the other hand, the momenta of the virtual vertices are fixed to zero, with no contribution from $ F_{1} \lambda_{1} $.\\
	
	Let us switch to the adapted coordinates. We can straightforwardly calculate the transformation matrices given as in \eqref{Wmatrices}, obtaining
	\begin{equation}\label{dotW0ex2}
	\dot{W}_{0} = \frac{1}{\sqrt{6}} \begin{pmatrix}
	2 \sqrt{2} & 2 \sqrt{2} & 2 \sqrt{2} & \sqrt{2} & \sqrt{2} & \sqrt{2}\\
	-\sqrt{3}/2 & 0 & \sqrt{3}/2 & -\sqrt{3} &  0 & \sqrt{3}\\
	-1/2 & 1 & -1/2 & -1 & 2 & -1 \\
	- \sqrt{2} &  - \sqrt{2} & - \sqrt{2} & 0 & 0 & 0\\
	\sqrt{3} &  0 & - \sqrt{3} & 0 & 0 & 0\\
	1 &  -2 & 1 & 0 & 0 & 0\\
	\end{pmatrix}
	\end{equation}
	and
	\begin{equation}\label{ddotW1ex2}
	\ddot{W}_{1} = \frac{1}{2 \sqrt{3}} \begin{pmatrix}
	0 & 12 & 0 & 0 & 0 & 0 \\
	0 & 0 & 2 \sqrt{3} & 0 & 0 & 0 \\
	2 \sqrt{3} & 0 & 0 & 0 & 0 & 0 \\
	0 & -6 & 0 & 0 & 1 & 0  \\
	0 & 0 & 0 & 0 & 0 & 2 \sqrt{3} \\
	0 & 0 & 0 & 2 \sqrt{3} & 0 & 0
	\end{pmatrix}
	\end{equation}
	If one needs the inversed versions of these matrices, which give explicitly the adapted bases (in language of the canonical ones), one can easily take the inverse by using \eqref{W-1}. Recalling Observations \ref{ob:ad1} and \ref{ob:ad2}, we can classify state-space vectors on both time-slices based on their adapted coordinates. Thus we know that $ \dot{y}_{0} $ belongs to the pre-constraint surface $ \mathcal{C}_{0}^{-} $ if and only if its second and third component are zero. The first and fourth component of $ \dot{y}_{0} $ represent field values and momenta, respectively, of vectors in the representative space $ \dot{\mathcal{C}}_{0}^{-} $. The fifth and sixth component of $ \dot{y}_{0} $ parametrize the null space $ \mathcal{N}_{\omega}(\mathcal{C}_{0}^{-}) $, which is a subspace of the pre-constraint surface $ \mathcal{C}_{0}^{-} $.\\
	
	For example, the vector $ y_{0} = \begin{pmatrix} 2 & 0 & 0 & -1 & 0 & 0 \end{pmatrix}^{T} $ clearly satisfies $ C_{0} y_{0} = 0 $, and thus belongs to the pre-constraint surface. Its adapted coordinates are $ \dot{y}_{0} = \dot{W}_{0} y_{0} = \frac{1}{\sqrt{6}} \begin{pmatrix} 3 \sqrt{2} & 0 & 0 & -2 \sqrt{2} & 2 \sqrt{3} & 2 \end{pmatrix}^{T} $. We see that this form confirms that $ y_{0} \in \mathcal{C}_{0}^{-} $. Moreover, it tells us that the vector has nonzero intersection with the null space $ \mathcal{N}_{\omega}(\mathcal{C}_{0}^{-}) $, and therefore does not belong to the representative space $ \dot{\mathcal{C}}_{0}^{-} $. Nevertheless, we can easily evolve it to time-slice 1 by using \eqref{uddotX}---\eqref{uddotMq} with the result $ \ddot{y}_{1} = \frac{1}{\sqrt{6}} \begin{pmatrix} 3 \sqrt{2} & \lambda_{1 1} & \lambda_{1 2} & -2 \sqrt{2} & 0 & 0 \end{pmatrix}^{T} $. Note that because of our choice of the adapted coordinates, the null space $ \mathcal{N}_{\omega}(\mathcal{C}_{1}^{+}) $ is parametrized by the second and third coordinate of $ \ddot{y}_{1} $ (which take the free parameters $ \lambda_{1 1}, \lambda_{1 2} $), while the complement $ \mathcal{P}_{1} \smallsetminus \mathcal{C}_{1}^{+} $ of the post-constraint surface to the full phase space is parametrized by the fifth and the sixth component. The vector $ y_{1} $ automatically belongs to the post-constraint surface $ \mathcal{C}_{1}^{+} $, thus the two zeros at these positions. The representative-space components are conserved. One can easily check that $ \ddot{y}_{1} = \ddot{W}_{1} y_{1} $ with $ y_{1} $ computed from \eqref{E0F0ex2}.

\end{example}

\vspace{\baselineskip}

\begin{example}\label{ex:3}
	Now we provide the third promised instance of a triangular one-step lattice, which is the time-reversed version of that from Example \ref{ex:2}. Its depiction is given in Fig. \ref{fig:lat3}. We again expect to find the system irregular, but since degrees of freedom are added, the irregularity should give rise to a nontrivial space of free parameters.
	
	\begin{figure}[h!]
		\centering
		\begin{tikzpicture}[scale=1]
		\tikzset{
			vertex/.style={
				shape=circle,fill=lightgray!100,minimum size=3mm,inner sep=0.2mm, label={[fill=none,label distance=1mm]90:#1}
			},
			vertexwhite/.style={
				shape=circle,fill=white!100,minimum size=10mm,inner sep=0.2mm, label={[fill=none,label distance=1mm]90:#1}
			},
			vertexvirtual/.style={
				shape=circle,draw=lightgray!100,fill=white!100,line width=0.3mm,minimum size=2.7mm,inner sep=0.2mm, 	label={[fill=none,label distance=1mm]90:#1}
			},
			edge/.style={
				draw,-,color=lightgray!100,line width=0.3mm
			},
			edget/.style={
				draw,dashed,color=lightgray!100,line width=0.3mm
			},
		}
		
		\coordinate (cia) at (-2,0);
		\coordinate (c1) at (-1,0);
		\coordinate (c2) at (0,0);
		\coordinate (c3) at (1,0);
		\coordinate (cib) at (2,0);
		\coordinate (cic) at (-2,0.866);
		\coordinate (c4) at (-1,0.866);
		\coordinate (c5) at (0,0.866);
		\coordinate (c6) at (1,0.866);
		\coordinate (cid) at (2,0.866);
		
		\draw[edget] (c4) -- (c2) -- (c6);
		\draw[edget] (c5) -- (c2);
		\draw[edge] (cic) -- (c4) -- (c5) -- (c6) -- (cid);
		
		\node[vertexvirtual] at (c1) {1};
		\node[vertex] at (c2) {2};
		\node[vertexvirtual] at (c3) {3};
		\node[vertex] at (c4) {4};
		\node[vertex] at (c5) {5};
		\node[vertexwhite] at (cia) {};
		\node[vertexwhite] at (cib) {};
		\node[vertexwhite] at (cic) {};
		\node[vertexwhite] at (cid) {};
		\node[vertex] at (c6) {6};
		\node[] at (3.5,0) {$ n = 0 $};
		\node[] at (3.5,0.866) {$ n = 1 $};
		\end{tikzpicture}
		\vspace{0 mm}
		\caption{Diagram of the time-step lattice of Example \ref{ex:3}. It is again made of three identical triangles, this time of type 1-2. Vertices 1 and 3 are virtual.}
		\label{fig:lat3}
		\vspace{4 mm}
	\end{figure}
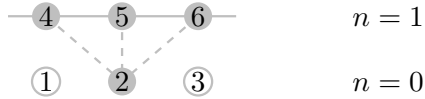

	We have
	\begin{equation}\label{key}
	K = \begin{pmatrix}
	0 & 0 & 0 & 0 & 0 & 0 \\
	& -6 & 0 & 2 & 2 & 2 \\
	& & 0 & 0 & 0 & 0 \\
	& & & -1 & -1/2 & -1/2 \\
	& & & & -1 & -1/2 \\
	& & & & & -1 \\
	\end{pmatrix}
	\end{equation}
	and thus
	\begin{equation}\label{key}
	L_{0} = \begin{pmatrix}
	0 & 0 & 0 \\
	0 & -6 & 0 \\
	0 & 0 & 0
	\end{pmatrix}, \qquad R_{0} = \begin{pmatrix}
	0 & 0 & 0 \\
	2 & 2 & 2 \\
	0 & 0 & 0
	\end{pmatrix}, \qquad \bar{R}_{1} = \frac{1}{2} \begin{pmatrix}
	2 & 1 & 1 \\
	1 & 2 & 1 \\
	1 & 1 & 2
	\end{pmatrix}
	\end{equation}
	The computation is fully analogical to the preceding case. We find
	\begin{equation}\label{key}
	U = \begin{pmatrix}
	0 & 0 & 1 \\
	1 & 0 & 0 \\
	0 & 1 & 0
	\end{pmatrix}, \qquad \Sigma = \begin{pmatrix}
	2 \sqrt{3} & 0 & 0 \\
	0 & 0 & 0 \\
	0 & 0 & 0
	\end{pmatrix}, \qquad V = \frac{1}{\sqrt{6}} \begin{pmatrix}
	\sqrt{2} & -\sqrt{3} & -1 \\
	\sqrt{2} & 0 & 2 \\
	\sqrt{2} & \sqrt{3} & -1
	\end{pmatrix}
	\end{equation}
	which gives
	\begin{equation}\label{key}
	C_{0} = \begin{pmatrix}
	0 & 0 & 0 & 1 & 0 & 0 \\
	0 & 0 & 0 & 0 & 0 & 0 \\
	0 & 0 & 0 & 0 & 0 & 1
	\end{pmatrix}
	\end{equation}
	There is a nontrivial pre-constraint, even though the number of physical degrees of freedom increases. The pre-constraint surface $ \mathcal{C}_{0}^{-} $ has dimension four---the same as in the previous example. However, observe that the pre-constraint is only concerned with the momenta at virtual vertices. It reflects the fact that the field at virtual vertices, however it looks, does not propagate to the future, therefore its momenta must be zero. The evolution of vectors $ y_{0} \in \mathcal{C}_{0}^{-} $ is described by
	\begin{equation}\label{E0F0ex3}
	E_{0} = \frac{1}{6} \begin{pmatrix}
	0 & 6 & 0 & 0 & 1 & 0 \\
	0 & 6 & 0 & 0 & 1 & 0 \\
	0 & 6 & 0 & 0 & 1 & 0 \\
	0 & 0 & 0 & 0 & 2 & 0 \\
	0 & 0 & 0 & 0 & 2 & 0 \\
	0 & 0 & 0 & 0 & 2 & 0
	\end{pmatrix}, \qquad F_{1} = \frac{1}{2 \sqrt{6}} \begin{pmatrix}
	-2 \sqrt{3} & -2 \\
	0 & 4 \\
	2 \sqrt{3} & -2 \\
	- \sqrt{3} & -1 \\
	0 & 2 \\
	\sqrt{3} & -1
	\end{pmatrix}
	\end{equation}
	As one expects, the evolution only takes into account the variables at vertex 2, the field values and momenta at virtual vertices 1 and 3 are irrelevant. This time, the free parameters $ \lambda_{1 1}, \lambda_{1 2} $ will influence all resulting field values and momenta.\\
	
	We can go on to the adapted bases. The transformation matrices are
	\begin{equation}\label{dotW0ex3}
	\dot{W}_{0} = \begin{pmatrix}
	0 & 6 & 0 & 0 & 1 & 0 \\
	0 & 0 & 0 & 0 & 0 & 1 \\
	0 & 0 & 0 & 1 & 0 & 0 \\
	0 &  - 1 & 0 & 0 & 0 & 0\\
	0 &  0 & - 1 & 0 & 0 & 0\\
	-1 &  0 & 0 & 0 & 0 & 0\\
	\end{pmatrix}
	\end{equation}
	and
	\begin{equation}\label{ddotW1ex3}
	\ddot{W}_{1} = \frac{1}{6 \sqrt{6}} \begin{pmatrix}
	12 \sqrt{6} & 12 \sqrt{6} & 12 \sqrt{6} & 0 & 0 & 0 \\
	- 6 \sqrt{3} & 0 & 6 \sqrt{3} & 0 & 0 & 0 \\
	-6 & 12 & -6 & 0 & 0 & 0 \\
	-2 \sqrt{6} & -2 \sqrt{6} & -2 \sqrt{6} & \sqrt{6} & \sqrt{6} & \sqrt{6}  \\
	3 \sqrt{3} & 0 & - 3 \sqrt{3} & -6 \sqrt{3} & 0 & 6 \sqrt{3} \\
	3 & -6 & 3 & -6 & 12 & -6
	\end{pmatrix}
	\end{equation}
	The classification of vectors in adapted coordinates is the same as before. To try it out, take for example the initial vector $ y_{0} = \begin{pmatrix} 0 & 1 & 0 & 0 & 1 & 0 \end{pmatrix}^{T} $. It clearly satisfies the pre-constraint $ C_{0} y_{0} = 0 $. In the adapted coordinates, it looks like $ \dot{y}_{0} = \dot{W}_{0} y_{0} = \begin{pmatrix} 7 & 0 & 0 & -1 & 0 & 0 \end{pmatrix}^{T} $. The zeros at positions two and three confirm that $ y_{0} \in \mathcal{C}_{0}^{-} $. The zeros at positions five and six are the result of our arbitrary choice, and mean that $ y_{0} $ has zero intersection with the null space $ \mathcal{N}_{\omega}(\mathcal{C}_{0}^{-}) $ and so $ y_{0} \in \dot{\mathcal{C}}_{0}^{-} $. Let us evolve this vector to time-slice 1. According to our trivial evolution prescription, we write $ \ddot{y}_{1} = \begin{pmatrix} 7 & \lambda_{1 1} & \lambda_{1 2} & -1 & 0 & 0 \end{pmatrix}^{T} $. As always, it holds $ \ddot{y}_{1} = \ddot{W}_{1} y_{1} $. The last two components of $ \ddot{y}_{1} $ tell us that we are in the post-constraint surface $ \mathcal{C}_{1}^{+} $. There are two free parameters entering the evolution just as in Example \ref{ex:2}. However, unlike in Example \ref{ex:2}, the present free parameters have physical significance, since they contribute to the field values and momenta at real vertices 4, 5 and 6 of the lattice.

\end{example}

\vspace{\baselineskip}

\begin{example}\label{ex:4}
	Eventually, let us consider a one-step lattice analogical to the lattice of Example \ref{ex:1}, but with only two vertices per time-slice. The diagram is given in Figure \ref{fig:lat4}. Because of the smaller number of vertices, the spacelike edges at both time-slices are doubled (we keep two edges between the two vertices of each time-slice in order to satisfy our assumption that each time-slice is a closed loop), and every vertex shares at least one edge with every other. This makes the lattice slightly unusual; nevertheless, it still formally complies to our assumptions.\\

	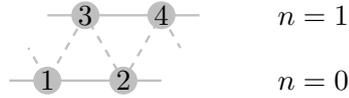
\begin{figure}[h!]
		\centering
		\begin{tikzpicture}[scale=1]
		\tikzset{
			vertex/.style={
				shape=circle,fill=lightgray!100,minimum size=3mm,inner sep=0.2mm, label={[fill=none,label distance=1mm]90:#1}
			},
			vertexwhite/.style={
				shape=circle,fill=white!100,minimum size=10mm,inner sep=0.2mm, label={[fill=none,label distance=1mm]90:#1}
			},
			edge/.style={
				draw,-,color=lightgray!100,line width=0.3mm
			},
			edget/.style={
				draw,dashed,color=lightgray!100,line width=0.3mm
			},
		}
		
		\coordinate (cia) at (-2,0);
		\coordinate (c1) at (-1,0);
		\coordinate (c2) at (0,0);
		\coordinate (cib) at (1,0);
		\coordinate (cic) at (-1.5,0.866);
		\coordinate (c4) at (-0.5,0.866);
		\coordinate (c5) at (0.5,0.866);
		\coordinate (cid) at (1.5,0.866);

		\draw[edge] (cia) -- (c1) -- (c2) -- (cib);
		\draw[edget] (cic) -- (c1) -- (c4) -- (c2) -- (c5) -- (cib);
		\draw[edge] (cic) -- (c4) -- (c5) -- (cid);
		
		\node[vertex] at (c1) {1};
		\node[vertex] at (c2) {2};
		\node[vertex] at (c4) {3};
		\node[vertex] at (c5) {4};
		\node[vertexwhite] at (cia) {};
		\node[vertexwhite] at (cib) {};
		\node[vertexwhite] at (cic) {};
		\node[vertexwhite] at (cid) {};
		\node[] at (2.5,0) {$ n = 0 $};
		\node[] at (2.5,0.866) {$ n = 1 $};

		
		\end{tikzpicture}
		\vspace{0 mm}
		\caption{Diagram of the time-step lattice of Example \ref{ex:4}. It is formed by four triangles (one doubled type 2-1 triangle and one doubled type 1-2 triangle).}
		\label{fig:lat4}
		\vspace{4 mm}
	\end{figure}
	
	Let us work out the corresponding matrices. We can take $ t = 1 $, $ q = 2 $, and so $ N = 4 $. The dynamical matrix is
	\begin{equation}\label{key}
	K = \begin{pmatrix}
	-3 & -1 & 2 & 2 \\
	& -3 & 2 & 2  \\
	& & -3 & -1 \\
	& & & -3
	\end{pmatrix}
	\end{equation}
	We have implemented the double edges simply by summing up the weights. It holds
	\begin{equation}\label{key}
	L_{0} = - \begin{pmatrix}
	3 & 1  \\
	1 & 3
	\end{pmatrix}, \qquad R_{0} = \begin{pmatrix}
	2 & 2 \\
	2 & 2
	\end{pmatrix}, \qquad \bar{R}_{1} = \begin{pmatrix}
	3 & 1  \\
	1 & 3
	\end{pmatrix}
	\end{equation}
	Now the catch is clear: the matrix $ R_{0} $ is not regular, instead $ r_{0} = 1 $ and $ s_{0} = 1 $. The singular value decomposition of $ R_{0} $ results in
	\begin{equation}\label{key}
	U = \frac{1}{\sqrt{2}} \begin{pmatrix}
	1 & -1 \\
	1 & 1 
	\end{pmatrix}, \qquad \Sigma = \begin{pmatrix}
	4 & 0 \\
	0 & 0 
	\end{pmatrix}, \qquad V = \frac{1}{\sqrt{2}} \begin{pmatrix}
	1 & -1 \\
	1 & 1
	\end{pmatrix}
	\end{equation}
	(note that $ U = V $, this is because $ R_{0} $ is symmetric) which gives
	\begin{equation}\label{key}
	C_{0} = \frac{1}{2} \begin{pmatrix}
	2 & -2 & 1 & -1 \\
	-2 & 2 & -1 & 1
	\end{pmatrix}
	\end{equation}
	The pre-constraint surface $ \mathcal{C}_{0}^{-} $ is not the whole $ \mathcal{P}_{0} $, it has dimension three. For vectors $ y_{0} $ in $ \mathcal{C}_{0}^{-} $, the evolution is fixed by
	\begin{equation}\label{E0F0ex4}
	E_{0} = \frac{1}{8} \begin{pmatrix}
	4 & 4 & 1 & 1 \\
	4 & 4 & 1 & 1 \\
	0 & 0 & 4 & 4 \\
	0 & 0 & 4 & 4 
	\end{pmatrix}, \qquad F_{1} = \frac{1}{\sqrt{2}} \begin{pmatrix}
	-1 \\
	1 \\
	-2 \\
	2
	\end{pmatrix}
	\end{equation}
	Looking at $ F_{1} $, one can see that the dimension of the null space $ \mathcal{N}_{\omega}(\mathcal{C}_{1}^{+}) $ is one, i.e., there is one free parameter $ \lambda_{1 1} $ of the evolution. The transformation matrices to the adapted bases are
	\begin{equation}\label{dotW0ex4ddotW1ex4}
	\dot{W}_{0} = \frac{1}{\sqrt{2}} \begin{pmatrix}
	4 & 4 & 1 & 1  \\
	-2 & -2 & -1 & 1  \\
	-1 & -1 & 0 & 0  \\
	1 & - 1 & 0 & 0
	\end{pmatrix}, \qquad \ddot{W}_{1} = \frac{1}{\sqrt{2}} \begin{pmatrix}
	4 & 4 & 0 & 0  \\
	-1 & -1 & 0 & 0  \\
	-1 & -1 & 1/4 & 1/4  \\
	2 & -2 & -1 & 1
	\end{pmatrix}
	\end{equation}
	Take for example the vector $ y_{0} = \begin{pmatrix} 1 & 0 & 0 & 2 \end{pmatrix}^{T} $, which satisfies the pre-constraint. Its adapted coordinates are $ \dot{y}_{0} = \frac{1}{\sqrt{2}} \begin{pmatrix} 5 & 0 & -1 & 1 \end{pmatrix}^{T} $. The zero at position two signifies that we are on the pre-constraint surface. The other three components parametrize the pre-constraint surface. In particular, the last component parametrizes $ \mathcal{N}_{\omega}(\mathcal{C}_{0}^{-}) $. Evolving to time-slice 1, we write $ \ddot{y}_{1} = \frac{1}{\sqrt{2}} \begin{pmatrix} 5 & \lambda_{1 1} & -1 & 0 \end{pmatrix}^{T} $. All this is a standard use of the formalism. It shows us that in spite of a constant number of degrees of freedom and high connectivity, the system we obtain is irregular. We interpret this behavior by saying that the lattice is \textit{overconnected}.

\end{example}

\vspace{\baselineskip}

\subsection{Multiple Time-Steps}

Eventually let us briefly comment on lattices with multiple time-steps. There is really nothing new to these, since they are but individual time-step lattices stacked on top of each other, forming a system arbitrarily extended in time. The evolution of such system is given simply as a series of the individual one-step evolution moves. Things can get more complicated if one asks questions about global properties of the system, e.g. when one wants to find initial data $ y_{0} $ which give rise to a global solution. In that case, one needs to trace back all the pre-constraints arising anywhere in the lattice. This is why we say that \textit{constraints propagate in time}, both to future and past. On the other hand, given a vector $ y_{0} \in \mathcal{P}_{n} $, one may evolve it by a series of local one-step evolution moves and in this way find its later versions. The solution (if it exists) may branch out with an increasing number of free parameters or tail off (meaning that it is restricted or even ceases to exist) due to pre-constraints. At all cases, we know well that the symplectic structure of solutions will be conserved in time.\\

To illustrate some of the possible behavior of multistep systems, we offer two examples. Both are built up from the time-step lattices of Examples \ref{ex:1}---\ref{ex:3} and extend over three time-steps. In other aspects, they are quite different.\\

\begin{example}
	The first case is depicted in Fig. \ref{fig:multilat1}. It starts and ends with a single vertex, but widens in between. Nevertheless, this widening has little effect on the propagating degrees of freedom, since the free parameters arising during time-step between 0 and 1 will be eventually diminished by the pre-constraint at time-slice 2.\\

	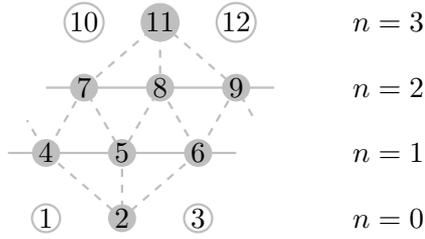
\begin{figure}[h!]
		\centering
		\begin{tikzpicture}[scale=1]
		\tikzset{
			vertex/.style={
				shape=circle,fill=lightgray!100,minimum size=3mm,inner sep=0.2mm, label={[fill=none,label distance=1mm]90:#1}
			},
			vertexwhite/.style={
				shape=circle,fill=white!100,minimum size=10mm,inner sep=0.2mm, label={[fill=none,label distance=1mm]90:#1}
			},
			vertexvirtual/.style={
				shape=circle,draw=lightgray!100,fill=white!100,line width=0.3mm,minimum size=2.7mm,inner sep=0.2mm, 	label={[fill=none,label distance=1mm]90:#1}
			},
			edge/.style={
				draw,-,color=lightgray!100,line width=0.3mm
			},
			edget/.style={
				draw,dashed,color=lightgray!100,line width=0.3mm
			},
		}
		
		\coordinate (cia) at (-2,0);
		\coordinate (c1) at (-1,0);
		\coordinate (c2) at (0,0);
		\coordinate (c3) at (1,0);
		\coordinate (cib) at (2,0);
		
		\coordinate (cic) at (-2,0.866);
		\coordinate (c4) at (-1,0.866);
		\coordinate (c5) at (0,0.866);
		\coordinate (c6) at (1,0.866);
		\coordinate (cid) at (2,0.866);
		
		\coordinate (cie) at (-1.5,1.732);
		\coordinate (c7) at (-0.5,1.732);
		\coordinate (c8) at (0.5,1.732);
		\coordinate (c9) at (1.5,1.732);
		\coordinate (cif) at (2.5,1.732);	
		
		\coordinate (cig) at (-1.5,2.598);
		\coordinate (c10) at (-0.5,2.598);
		\coordinate (c11) at (0.5,2.598);
		\coordinate (c12) at (1.5,2.598);
		\coordinate (cih) at (2.5,2.598);

		\draw[edget] (c4) -- (c2) -- (c6);
		\draw[edget] (c5) -- (c2);
		\draw[edge] (cic) -- (c4) -- (c5) -- (c6) -- (cid);
		\draw[edget] (cie) -- (c4) -- (c7) -- (c5) -- (c8) -- (c6) -- (c9) -- (cid);
		\draw[edge] (cie) -- (c7) -- (c8) -- (c9) -- (cif);
		\draw[edget] (c7) -- (c11) -- (c8);
		\draw[edget] (c11) -- (c9);
		
		\node[vertexvirtual] at (c1) {1};
		\node[vertex] at (c2) {2};
		\node[vertexvirtual] at (c3) {3};
		\node[vertex] at (c4) {4};
		\node[vertex] at (c5) {5};
		\node[vertex] at (c6) {6};
		\node[vertex] at (c7) {7};
		\node[vertex] at (c8) {8};
		\node[vertex] at (c9) {9};
		\node[vertexvirtual] at (c10) {10};
		\node[vertex] at (c11) {11};
		\node[vertexvirtual] at (c12) {12};
		
		\node[vertexwhite] at (cia) {};
		\node[vertexwhite] at (cib) {};
		\node[vertexwhite] at (cic) {};
		\node[vertexwhite] at (cid) {};
		\node[vertexwhite] at (cie) {};
		\node[vertexwhite] at (cif) {};
		\node[vertexwhite] at (cig) {};
		\node[vertexwhite] at (cih) {};
		
		\node[] at (3.5,0) {$ n = 0 $};
		\node[] at (3.5,0.866) {$ n = 1 $};
		\node[] at (3.5,1.732) {$ n = 2 $};
		\node[] at (3.5,2.598) {$ n = 3 $};
		
		\end{tikzpicture}
		\vspace{0 mm}
		\caption{Diagram of a lattice with a temporal widening.}
		\label{fig:multilat1}
		\vspace{4 mm}
	\end{figure}

\end{example}

\begin{example}
	The second case of a multistep system is depicted in Fig. \ref{fig:multilat2}. This time the lattice begins and ends with three vertices which carry three degrees of freedom. The evolution between time-slices 0 and 1 is regular, but at time-slice 2 the lattice narrows down to a single vertex, thus obstructing the propagation of degrees of freedom. In result, the number of propagating degrees of freedom is restricted to one.\\
	
	\begin{figure}[h!]
		\centering
		\begin{tikzpicture}[scale=1]
		\tikzset{
			vertex/.style={
				shape=circle,fill=lightgray!100,minimum size=3mm,inner sep=0.2mm, label={[fill=none,label distance=1mm]90:#1}
			},
			vertexwhite/.style={
				shape=circle,fill=white!100,minimum size=10mm,inner sep=0.2mm, label={[fill=none,label distance=1mm]90:#1}
			},
			vertexvirtual/.style={
				shape=circle,draw=lightgray!100,fill=white!100,line width=0.3mm,minimum size=2.7mm,inner sep=0.2mm, 	label={[fill=none,label distance=1mm]90:#1}
			},
			edge/.style={
				draw,-,color=lightgray!100,line width=0.3mm
			},
			edget/.style={
				draw,dashed,color=lightgray!100,line width=0.3mm
			},
		}
		
		\coordinate (cia) at (-2,0);
		\coordinate (c1) at (-1,0);
		\coordinate (c2) at (0,0);
		\coordinate (c3) at (1,0);
		\coordinate (cib) at (2,0);
		
		\coordinate (cic) at (-1.5,0.866);
		\coordinate (c4) at (-0.5,0.866);
		\coordinate (c5) at (0.5,0.866);
		\coordinate (c6) at (1.5,0.866);
		\coordinate (cid) at (2.5,0.866);
		
		\coordinate (cie) at (-1.5,1.732);
		\coordinate (c7) at (-0.5,1.732);
		\coordinate (c8) at (0.5,1.732);
		\coordinate (c9) at (1.5,1.732);
		\coordinate (cif) at (2.5,1.732);	
		
		\coordinate (cig) at (-1.5,2.598);
		\coordinate (c10) at (-0.5,2.598);
		\coordinate (c11) at (0.5,2.598);
		\coordinate (c12) at (1.5,2.598);
		\coordinate (cih) at (2.5,2.598);

		\draw[edge] (cia) -- (c1) -- (c2) -- (c3) -- (cib);
		\draw[edget] (cic) -- (c1) -- (c4) -- (c2) -- (c5) -- (c3) -- (c6) -- (cib);
		\draw[edge] (cic) -- (c4) -- (c5) -- (c6) -- (cid);
		\draw[edget] (c4) -- (c8) -- (c5);
		\draw[edget] (c8) -- (c6);
		\draw[edget] (c10) -- (c8) -- (c11);
		\draw[edget] (c8) -- (c12);
		\draw[edge] (cig) -- (c10) -- (c11) -- (c12) -- (cih);

		\node[vertex] at (c1) {1};
		\node[vertex] at (c2) {2};
		\node[vertex] at (c3) {3};
		\node[vertex] at (c4) {4};
		\node[vertex] at (c5) {5};
		\node[vertex] at (c6) {6};
		\node[vertexvirtual] at (c7) {7};
		\node[vertex] at (c8) {8};
		\node[vertexvirtual] at (c9) {9};
		\node[vertex] at (c10) {10};
		\node[vertex] at (c11) {11};
		\node[vertex] at (c12) {12};
		
		\node[vertexwhite] at (cia) {};
		\node[vertexwhite] at (cib) {};
		\node[vertexwhite] at (cic) {};
		\node[vertexwhite] at (cid) {};
		\node[vertexwhite] at (cie) {};
		\node[vertexwhite] at (cif) {};
		\node[vertexwhite] at (cig) {};
		\node[vertexwhite] at (cih) {};
		
		\node[] at (3.5,0) {$ n = 0 $};
		\node[] at (3.5,0.866) {$ n = 1 $};
		\node[] at (3.5,1.732) {$ n = 2 $};
		\node[] at (3.5,2.598) {$ n = 3 $};
		
		\end{tikzpicture}
		\vspace{0 mm}
		\caption{Diagram of a lattice with a temporal narrowing.}
		\label{fig:multilat2}
		\vspace{4 mm}
	\end{figure}
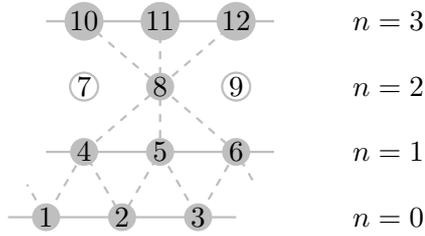

\end{example}

\vspace{\baselineskip}

\section{Conclusion}
In this work, the existing formalism of discrete canonical evolution was revisited and applied to the case of linear dynamical system, i.e., a system with vector configuration space and quadratic action. Thanks to the very strong assumption of linearity, we could rewrite the one-step evolution into a simple matrix form. The key object in this formulation is the matrix $ R_{n} $ describing interaction of variables between time-steps $ n $ and $ n+1 $. One can easily obtain the explicit Hamiltonian evolution map, all it takes is a singular value decomposition of $ R_{n} $. For an irregular system, the evolution map is only defined on a subset $ \mathcal{C}_{n}^{-} $ of the phase space $ \mathcal{P}_{n} $ called the pre-constraint surface and is neither unique nor symplectic.\\

In order to understand the symplectic structure of the model, we performed an analysis of the constraint surfaces in relation to the symplectic form. Then we constructed two special bases of the phase space $ \mathcal{P}_{n} $ at each discrete time $ n $ which explicitly separate the constraint surfaces, the null spaces and the subspaces of propagating degrees of freedom. The corresponding transformations were given in terms of two symplectic matrices \eqref{dotW} and \eqref{ddotWorig}. Thanks to this construction, we were able to introduce a reduced evolution map $ \mathbb{H}_{n+1}(0): \dot{\mathcal{C}}_{n}^{-} \rightarrow \dot{\mathcal{C}}_{n+1}^{+} $ which \textit{is symplectic}. Moreover, it was shown that in the \textit{adapted coordinates} given by the new bases, the general one-step evolution map assumes a trivial form. We also gave some theoretical background for considering global solutions.\\

In comparison with a previously published article \cite{Hoehn2014} on the topic, we made a~number of decisions that lead to a significantly different approach. We gave much more attention to the symplectic structure since we consider it to be the most important object of the model. We limited our analysis of the Hamiltonian evolution to a~single time-step, which resulted in a less complex and arguably more straightforward treatment. Unlike in \cite{Hoehn2014}, the construction is performed explicitly in terms of matrices present in the action or their singular value decomposition. This allows for a smooth implementation and an easy application of the formalism to any problem dealing with a linear discretely evolving system.\\

In the last section we provided a fully worked-out example of discrete linear evolution of massless scalar field on a fixed two-dimensional spacetime lattice. Although a toy model, it has a sound physical base, and one can easily think of generalizations. One can observe how the scalar field is shaped by the lattice, shaped by its geometry and causal structure. The example is closely related to the intended application of the present work, which is the description of matter or gauge fields on a fixed spacetime lattice in the manner similar to (quantum) field theory on curved spacetime. With this example we also demonstrated in simple fashion the most important features of the irregular linear evolution as well as its overall utility, and illustrated the previously introduced formalism.\\

The present analysis is supposed to serve one more purpose, namely to provide the necessary tools for a subsequent treatment of a quantum version of the considered system. This is an interesting and relevant problem, addressed before in Sec. 10 of \cite{Hoehn2014} and more generally in the preceding work \cite{Hoehn2014a, Hoehn2014b}. The pursuit craves for special preparation since the standard quantization procedure typically requires a one-to-one symplectic evolution map on the phase space which can be used to induce the corresponding unitary evolution map on the Hilbert space describing the system. However, the case of \textit{discrete linear canonical evolution} does not meet this requirement. There are of course ways of surpassing this problem, yet they are easier to follow with an appropriate set of tools and good understanding of the classical design. Within this paper, we have spent effort to increase our understanding and prepare grounds for the following work concerning the quantum analogue. It is currently under preparation and likely to appear sometime in the near future.\\

\vspace{\baselineskip}

\section*{Acknowledgments}

This work was supported by Charles University Grant Agency [Project No. 906419].

\vspace{3\baselineskip}

\bibliographystyle{unsrt}
\renewcommand{\bibname}{Bibliography}
\bibliography{bibliography}

\end{document}